\documentclass[12pt]{article}

\usepackage{graphicx}
\usepackage{caption}
\usepackage{amsmath,amsfonts,bm,amssymb}
\usepackage{natbib}
\usepackage{array,dcolumn,booktabs,float}
\usepackage{amsfonts}
\usepackage{amssymb}  
\usepackage{amsthm}  
\usepackage{subfigure}
\usepackage{setspace}
\usepackage{verbatim}
\usepackage{color}
\usepackage{soul}
\usepackage{url}
\usepackage{hyperref}
\usepackage{qtree}
\usepackage{algorithm}
\usepackage{algorithmic}
\usepackage{pdflscape}

\newcommand{\R}{\mathbb{R}}

\newcommand{\I}{1{\hskip -2.5 pt}\hbox{I} }

\newcommand{\mockalph}[1]{}

\newtheorem{theorem}{Theorem}

\newcommand{\btheta}{ {\boldsymbol \theta} }

\newcommand{\bmu}{ {\boldsymbol \mu} }
\newcommand{\balpha}{ {\boldsymbol \alpha} }
\newcommand{\boeta}{ {\boldsymbol \eta} }

\newcommand{\minx}{0}
\newcommand{\maxx}{\infty}
\newcommand{\minpar}{0}
\newcommand{\maxpar}{\infty}
\newcommand{\comire}{\textsc{c}o\textsc{m}i\textsc{r}e }
\newcommand{\comirep}{\textsc{c}o\textsc{m}i\textsc{r}e}

\bibpunct{(}{)}{;}{a}{}{,}

\makeatletter
\renewcommand\section{\@startsection {section}{1}{\z@}%
                                   {-3.5ex \@plus -1ex \@minus -.2ex}%
                                   {2.3ex \@plus.2ex}%
                                   {\normalfont\large\scshape}}
                                   
\renewcommand\subsection{\@startsection {subsection}{1}{\z@}%
                                   {-3.5ex \@plus -1ex \@minus -.2ex}%
                                   {2.3ex \@plus.2ex}%
                                   {\normalfont}}                                   
\makeatother

\title{Convex Mixture Regression for Quantitative Risk Assessment}
\author{Antonio Canale\thanks{Department of Statistical Sciences, University of Padua \tt canale@stat.unipd.it}, $\,$
Daniele Durante\thanks{Department of Decision Sciences, Bocconi University \tt daniele.durante@unibocconi.it} $\,$ and 
David B. Dunson\thanks{Department of Statistical Science, Duke University \tt dunson@duke.edu}
}

\begin{document}

\maketitle
	
\abstract{
There is wide interest in studying how the distribution of a continuous response changes with a predictor.  We are motivated by environmental applications in which the predictor is the dose of an exposure and the response is a health outcome. A main focus in these studies is inference on dose levels associated with a given increase in risk relative to a baseline. Popular methods either dichotomize the continuous response or focus on modeling changes with the dose in the expectation of the outcome. Such choices may lead to information loss and provide inaccurate inference on dose--response relationships. We instead propose a Bayesian convex mixture regression model that allows the entire distribution of the health outcome to be unknown and changing with the dose.  To balance flexibility and parsimony, we rely on a mixture model for the density at the extreme doses, and express the conditional density at each intermediate dose via a convex combination of these extremal densities.  This representation generalizes classical dose--response models for quantitative outcomes, and provides a more parsimonious, but still powerful, formulation compared to nonparametric methods, thereby improving interpretability and efficiency in inference on risk functions. A Markov chain Monte Carlo algorithm for posterior inference is developed, and the benefits of our methods are outlined in simulations, along with a study on the impact of \textsc{ddt} exposure on gestational age.
}


\maketitle

\section{Introduction}\label{intro}
It is of substantial interest to study how the distribution of an outcome $y \in \mathcal{Y}$ varies with a predictor $x \in \mathcal{X} \subseteq \R$.  Focusing on the case in which $\mathcal{Y} \subseteq \R$, with the response variable univariate and continuous, we let $f(y {\mid} x)$ denote the conditional density of $y$ given $x$.  Our  focus is on environmental health studies in which $x$ is the dose of a potentially adverse exposure, and $y$ defines a health response.  In such studies there is a wide interest in relating dose to the risk of an adverse health outcome, with such dose--response models forming the basis of  quantitative risk assessment and regulatory guidelines on safe levels of exposure \citep[e.g.][Chapter 4]{pieg:2005}.

As a real--data illustrative application---which incorporates features common to different observational dose--response studies---we focus on relating \textsc{dde} exposure in pregnant women to the risk of a premature delivery \citep{longnecker2001}.  The data set is obtained from a sub-study of the US Collaborative Perinatal Project (CPP), which was a large prospective study of US pregnant women and their children.  In this sub-study, \textsc{dde}, which is a persistent metabolite of the pesticide  \textsc{ddt}, was measured in the maternal serum during the third trimester of pregnancy.  \textsc{dde} is lipophilic and is stored in fat tissues, so that women build up a body burden which could affect pregnancy.  

\begin{figure*} 
\centering
   \includegraphics[width=.98\textwidth]{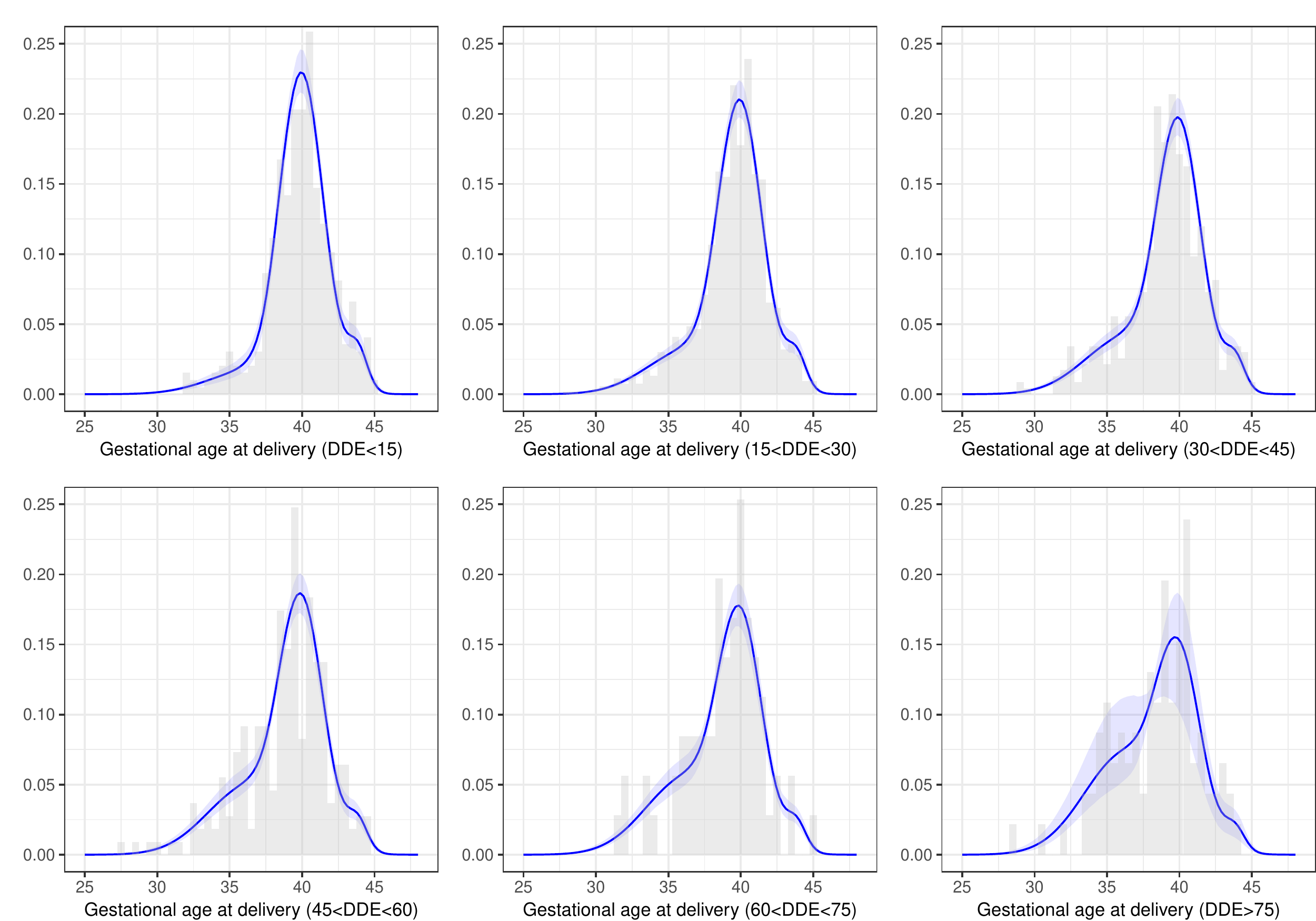} 
   \caption{Histograms of the observed gestational age at delivery for selected dose intervals, along with the pointwise posterior mean (continuous blue lines), and 95\% credible intervals (shaded blue areas) for $f(y\mid x)$ calculated in the central points  $(7.5, 22.5, 37.5, 52.5, 67.5, 125)$ of each dose interval, under the proposed convex mixture regression model. The figure appears in color in the electronic version of this article.}
   \label{fig:motivate}
\end{figure*}

Available  risk assessment studies for these  data rely on a Gaussian regression to infer changes  in gestational age with  \textsc{dde} exposure, or leverage a logistic regression to study  \textsc{dde} effects on preterm birth, which coincides, typically, with a gestational age  below $37$ weeks.  Although these models are useful, according to the histograms in Figure~\ref{fig:motivate}, the Gaussian assumption  is unrealistic, whereas dichotomizing the data prior to inference   leads to information loss \citep[e.g.][]{westkondel}. In fact, to obtain a more complete and accurate characterization of risk it is necessary to estimate how the whole distribution of the gestational age at delivery shifts with  \textsc{dde}. We address this goal via a novel Bayesian convex mixture regression (\comirep) model which allows the entire distribution of the gestational age to be unknown and changing with dose. The proposed formulation combines mixture representations and convex interpolations to balance flexibility and parsimony. This allows more accurate inference and quantitative risk assessment compared to Gaussian regression, and improves efficiency along with interpretability compared to fully nonparametric models.

\subsection{Models and methods for quantitative risk assessment}\label{lit}
In quantitative risk assessment one seeks to estimate the dose level associated with a given small increase in risk of an adverse health event relative to the background risk corresponding in general to zero dose.  This estimated  benchmark dose (BMD$_q$) is useful in supporting regulatory decisions, thus motivating inference on dose--response relationships to obtain the BMD$_q$---defined as the dose corresponding to a 100$q$\% risk of an adverse health event above that for unexposed individuals.

Due to the abundance of dichotomous health data, the vast majority of the literature has focused on discrete responses \citep[e.g.][]{piego_2014,fronc:2014}. There has been instead less consideration of continuous outcomes. Earlier approaches dichotomize the continuous response via  a clinical threshold and then apply models for binary data \citep[e.g.][]{longnecker2001}. From a statistical perspective, this  leads to a loss of information, efficiency, and validity compared to dose--response studies on  the original scale of the data \citep{crump:1995,westkondel}, thereby motivating a literature on quantitative risk assessment modeling the continuous response on its original scale. Popular formulations assume homoscedastic  Gaussian responses with a dose--dependent expectation $\mu(x)$ commonly defined as $\mu(x)={\mu_{\minpar}}+({\mu_{\maxpar}}-{\mu_{\minpar}})\psi(x; {\boldsymbol{\lambda}})$ for every $x \in [0, \infty)$, with $\mu_0$ and $\mu_\infty$ denoting the expectations of $y$ at $x=0$ and $x \to \infty$, respectively, whereas $\psi(x; {\boldsymbol{\lambda}})$ represents a monotone nondecreasing dose--response function having $\psi(0; {\boldsymbol{\lambda}})=0$, and $\psi(x; {\boldsymbol{\lambda}})\rightarrow 1$ for $x \rightarrow \infty$  \citep{ritz:2013,ritz_2005}. Other parametric models for $\mu(x)$ can be found in \citet{westkondel,kodell1993, piego:2005}.

Once the expectation $\mu(x)$ has been estimated, the excess risk above the background is commonly measured via the  additional risk function $R_{\mbox{\textsc{a}}}(x,a)$  \citep[][]{kodell1993}, which is defined as
\begin{eqnarray} 
R_{\mbox{\textsc{a}}}(x,a) &=&  \mbox{pr}(y \leq a \mid x)- \mbox{pr}(y \leq a \mid x=0) = F_x(a)-F_0(a),
\label{add_risk}
\end{eqnarray}
where $a \in \mathcal{Y}$ corresponds, in general, to a threshold of clinical interest, and $F_x(a)=F(a \mid x)$ is the conditional cumulative distribution function of the response at dose exposure $x$, computed in the threshold $a$. In the absence of relevant clinical guidelines for the cutoff value $a$, it is common practice to rely on tail values such as $a=\mu_0 - k \sigma$ with $k=2$ or $k=3$ \citep{kodell1993, westkondel, piego:2005}. \citet{yucatalano} consider instead quantile regression, which avoids specification of a cutoff  and incorporates heteroscedasticity in Gaussian data. Based on \eqref{add_risk}, the  BMD$_q$ of interest for inference is the solution of the equation $R_{\mbox{\textsc{a}}}(x,a) = q$. Hence, differently from dichotomization strategies, the above methods leverage the entire information in the data to learn $F_x(a)$, and require a cutoff---e.g. $a=37$ in our application---after the model has been estimated. According to  \eqref{add_risk}, reliable inference on $R_{\mbox{\textsc{a}}}(x,a)$ and BMD$_q$ requires accurate learning of $F_x(a)$, corresponding to an unknown continuous conditional distribution.

Although the aforementioned methods facilitate tractable inference on $F_x(a)$ and $R_{\mbox{\textsc{a}}}(x,a)$, the Gaussian assumption is difficult to justify in many toxicology studies. For example, as shown in Figure~\ref{fig:motivate}, this assumption is violated in our  illustrative data set, in terms of heteroscedasticity, skewness, and multimodality. Due to this, alternative contributions consider mixtures of Gaussians, which improve flexibility  while maintaining tractable  inference in risk assessments. Refer to \citet{razzaghi2000}, \citet{he:etal2010}, and \citet{hwan:pennel} for relevant formulations. These models rely on a general specification for the conditional density $f_x(y)=f(y\mid x)$, via 
\begin{eqnarray}
f_x(y)= \sum_{h=1}^H \nu_h(x)  \frac{1}{\sigma_h}\phi\left\{ \frac{y - \mu_h(x)}{\sigma_h}\right\},
\label{eq:mixtureregression}
\end{eqnarray}
where $\nu_h(x)$ is a probability weight for the $h$-th mixture component varying with $x$, and $\phi\{\cdot\}$ is the standard Gaussian density. 

In performing inference under the above mixture model, \citet{razzaghi2000} consider a representation with $H=2$ mixture components and fixed probability weights  $\nu_h(x) = \nu_h$ for all $x \in \mathcal{X}$. This choice is motivated by  the nonresponse phenomenon  in dose--response studies \citep[e.g.][]{good:1979}, but $H=2$ components may be insufficiently flexible in many situations.   An alternative possibility to address this issue is to fix $H$ at a conservative upper bound---potentially infinite---and allow the data to inform on the occupied number of mixture components under a Bayesian framework. One example  corresponds to the class of \textsc{anova} dependent Dirichlet process (\textsc{ddp}) mixture models developed by  \citet{deio:etal}. Although this formulation provides a more general alternative to a finite mixture model, there are still flexibility issues associated with the assumption of constant probability weights. This has motivated more flexible representations allowing the weights to change with $x$ \citep[e.g.][]{duns:park:2008, rodriguez:2011}. Notable applications in quantitative risk assessment inspired by these models can be found in \citet{he:etal2010}, and \citet{hwan:pennel}. However, the increased flexibility of these methods comes at a cost in terms of interpretability and parsimony in characterizing functionals of interest.  This may lead to a loss of efficiency, thereby increasing uncertainty in regulatory guidelines. 

One possibility to improve upon these methods is to incorporate further structure, balancing  interpretability, efficiency and flexibility. As described in Section   \ref{model}, we address this goal via  \comirep. Our formulation interpolates two extremal densities $f_0(y)$ and $f_{\infty}(y)$---defined via mixture models---through a single monotone increasing function $\beta(x) \in [0,1]$, which induces a flexible, yet interpretable and parsimonious, characterization of $f_x(y)$. Although inference for the statistical model outlined in Section \ref{model} can proceed under different paradigms---including the frequentist one---we focus here on a Bayesian implementation. This allows coherent uncertainty quantification and tractable inference on relevant quantities for quantitative risk assessment---such as the additional risk $R_{\mbox{\textsc{a}}}(x,a)$, the BMD$_q$, and more conservative benchmark doses BMDL$_q$ based on lower credible bounds of the BMD$_q$ \citep[e.g.][]{crump:1995, piego:2005}. These quantities are non-linear functions of the parameters, thereby requiring asymptotic results and approximations for tractable frequentist inference. Under a Bayesian approach, the posterior for these quantities can be  obtained via Monte Carlo methods exploiting the posterior samples of the parameters. Moreover, a Bayesian treatment allows formal incorporation of prior knowledge and inclusion of relevant restrictions. This Bayesian implementation is described in Section~\ref{gibbs}, and is compared against  notable competitors on several simulation experiments in Section \ref{simu}. The CPP application is explored in Section \ref{cpp}, whereas Section \ref{disc} contains concluding remarks and generalizations. 

The code, data, and tutorial implementation of \comire  are available at \url{github.com/tonycanale/CoMiRe}. Additional results and analyses can be found also in the Supplementary Materials.

\section{Convex Mixture Regression}
\label{model}
Let $y \in \mathcal{Y}\subseteq \R$. Our focus is on modeling the conditional density $y \mid x \sim f_x(y)$ via the mixture model
\begin{eqnarray}
f_{x}(y)=f(y \mid x)= \int_{\Theta} \mbox{K}(y; \btheta) d P_{x}(\btheta),
\label{mix}
\end{eqnarray}
where $\mbox{K}(y;\btheta)$ is a kernel with support on $\mathcal{Y}$ and parametrized by $\btheta \in \Theta$, whereas $P_x$ is a mixing measure changing with $x$. Consistent with our focus, let $\mathcal{X}= \R^+$. To clarify how \comire can be applied to many kernels, we avoid choosing a specific $\mbox{K}(y;\btheta)$ for the moment.

Recalling  Section \ref{lit}, we look for a specification of  $P_{x}$ which maintains a degree of flexibility in modeling $f_x(y)$, while facilitating interpretable and efficient inference on relevant functionals for quantitative risk assessment. We address this goal by assuming $P_x$ is a convex combination of the two mixing measures $P_{\minx}$ and $P_{\maxx}$ at $x=0$ and $x \rightarrow \maxpar$, respectively, obtaining 
\begin{eqnarray}
P_x=\{1-\beta(x)\}P_{\minx}+\beta(x)P_{\maxx},  \quad x>0,
\label{eq:convex1}
\end{eqnarray}
with $\beta(x)$ a  bounded monotone nondecreasing interpolation function having $\beta(0)=0$ and \mbox{$\beta(x) \rightarrow 1$} as $x \to \infty$. Hence, changes in $P_x$  with $x$ are controlled by a nondecreasing  function $\beta(x)$ inducing a continuous shift in $P_x$ from $P_{\minx}$ to $P_{\maxx}$, as $x$ grows. There are several ways to define $\beta(x)$ under the conditions of monotonicity---i.e $\beta'(x) \geq 0$---and codomain constraints---i.e. $\beta(x) \in [0,1]$.  A simple strategy to address this goal is to consider the basis expansion 
\begin{equation}
\beta(x) =  \sum_{j=1}^{J} w_j \psi_j(x), \quad x >0 ,
\label{eq:mixIsplines}
\end{equation}
where $\psi_1(x), \ldots, \psi_{J}(x)$ are monotone nondecreasing functions with 
\mbox{$\psi_{j}(x) \in [0,1]$} for all $j=1, \dots, J$. With this choice, the constraints on $\beta(x)$ are translated into \mbox{$0\leq w_j \leq1$} and \mbox{ $\sum_{j=1}^J w_j = 1$.} As will be clarified in Section \ref{gibbs}, this facilitates Bayesian inference. A tractable default choice for $\psi_1(x), \ldots, \psi_{J}(x)$, which is provably flexible in shape-constrained inference, is the I-splines basis \citep{ramsay:1988}. However, other specifications are possible, including general cumulative distribution functions or classical dose--response functions  \citep{ritz:2013,ritz_2005}.

To complete the specification of \comirep, we need a representation for the extreme mixing measures $P_{\minx}$ and $P_{\maxx}$. In defining these quantities, we aim to preserve flexibility, while incorporating additional relevant structure available in quantitative risk assessment studies.   In particular, the predictor is the dose of a chemical having increasingly adverse health effects at high values.  Hence, as $x$ grows, we expect the outcome to be centered on increasingly adverse health profiles, meaning that $\mbox{E}(y \mid x)=\mu(x)$ is monotone decreasing with $x$, when adverse health occurs for low values of $y$, as in our study.  We include this adversity profile property and flexibility in  $f_x(y)$, by letting
\begin{eqnarray}
P_{\minx} =  \sum_{h=1}^{H} \nu_{\minpar h} \delta_{\btheta_{\minpar h}},\quad  \quad
P_{\maxx} =  \delta_{\btheta_{\maxpar}},
\label{eq:convex2}
\end{eqnarray}
with $\delta_\btheta$ being the Dirac's delta mass at point $\btheta$, $H$ the total number of mixture components in $x=0$, and appropriate restrictions for the atoms, so that the single mixture component at $x\to \infty$ is centered on a more adverse health profile than all the mixture components at $x=0$, obtaining
\begin{eqnarray}
\mu_\maxpar < \mu_{\minpar h}, \quad \mbox{for every} \quad h=1, \ldots, H.
\label{eq:order_restriction}
\end{eqnarray}
In \eqref{eq:order_restriction}, ${\mu_{\maxpar}}$ denotes the expectation of the single mixture component at the extreme dose range, whereas  $\mu_{\minpar h}$ is the expectation of the mixture component $h$  characterizing the density at the low end of the dose range, for every $h=1, \ldots, H$. Leveraging representations  \eqref{mix}--\eqref{eq:convex2}, equation \eqref{eq:order_restriction} is sufficient to guarantee that $f_x(y)$ will shift to be centered on more adverse values of the health response as dose $x$ increases.  Indeed, under  \eqref{mix}--\eqref{eq:convex2}, we obtain
\begin{eqnarray}
\mbox{E}(y \mid x)=\mu(x)  &=& \{1-\beta(x)\}\mbox{E}(y \mid {\minpar})+\beta(x)\mbox{E}(y \mid {\maxpar})\notag  \\
&=& {\boldsymbol{\nu}}_{{\minpar}}^{\intercal}{\bmu_{\minpar}}+({\mu_{\maxpar}}-{\boldsymbol{\nu}}_{{\minpar}}^{\intercal}{\bmu_{\minpar}}) \beta(x), 
\label{mean}
\end{eqnarray}
where ${\boldsymbol{\nu}}_{{\minpar}}=(\nu_{{\minpar}1}, \ldots, \nu_{{\minpar}H})^{\intercal}$ and ${\bmu_{\minpar}}=({\mu_{\minpar1}}, \ldots, {\mu_{\minpar H}} )^{\intercal}$ are the vectors containing the mixing weights and the expectations, respectively, of the components at the low end.

Although it is possible to relax \eqref{eq:order_restriction} in various ways without affecting the monotonicity assumption for $\mu(x)$, assuming the  component at the high end of the dose range to be centered on a more adverse health profile than all the components at the low end has appealing interpretation in dose--response studies. Indeed,  equation \eqref{eq:order_restriction} implies that the single component at the very high end of $\mathcal{X}$ corresponds to the most adverse health profile. Under \eqref{eq:convex1} and \eqref{eq:convex2}, this profile has probability zero and $\beta(x)$, when $x=0$ and $x>0$, respectively. This assumption  interestingly incorporates common nonresponse effects  in toxicology studies  \citep{good:1979} by allowing the exposed individuals to have a different degree $\beta(x)$ of susceptibility to dose, which increases with $x$. However, differently from current formulations \citep[e.g.][]{razzaghi2000}, \comire improves flexibility in modeling the nonresponders density $f_0(y)$ via a mixture model accounting for possible adverse effects of unobserved chemicals. Since there is less information in the data to learn $f_{\infty}(y)$, due to sparsity at high dose exposures, the responders density is defined via a  single kernel. Although this choice seems restrictive, in practice $x\rightarrow \infty$ is  never observed and we do not attempt inference on $f_x(y)$ for very large $x$.  This leads to robustness to the assumptions in $x\rightarrow \infty$, with $f_x(y)$ fully flexible at low dose and shifting steadily towards being centered on more adverse values as $x$ grows. 

Generalizations to include more complex  $f_{\infty}(y)$ are possible. We attempted this extension in initial analyses but observed no evident improvements.  As we will outline in  Sections~\ref{simu} and \ref{cpp}, our basic \comire leads to accurate results, and therefore we avoid complications affecting interpretation. 

\subsection{Interpretation and properties}
\label{prop}

Model \eqref{mix}--\eqref{eq:convex1} has several interpretations worth exploring. A fundamental one for quantitative risk assessment relates $\beta(x)$ to the additional risk function $R_{\mbox{\textsc{a}}}(x,a)=  F_x(a) - F_{{\minx}}(a)$ via
\begin{eqnarray}
R_{\mbox{\textsc{a}}}(x,a)  &=& \{1 -\beta(x)\}F_{{\minx}}(a) + \beta(x)F_{{\maxx}}(a) - F_{{\minx}}(a) \notag\\
&=& \beta(x) \{F_{{\maxx}}(a) - F_{{\minx}}(a)\} \notag\\
&=& \beta(x)R_{\mbox{\textsc{a}}}(\infty,a),  
\label{eq:risk}
\end{eqnarray}
where $F_{0}(a) =\mbox{pr}(y \leq a \mid 0)$, and $F_{\infty}(a) =\mbox{pr}(y \leq a \mid \infty)$. According to \eqref{eq:risk}, the benefit associated with our representation, compared to the unstructured density regression models discussed in Section~\ref{lit}, is that the effect of $x$ enters only via the flexible function $\beta(x)$ which is proportional to $R_{\mbox{\textsc{a}}}(x,a)$. Indeed, consistent with  \eqref{eq:risk}, the dose--response function $\beta(x)$ can be interestingly interpreted as a standardized additional risk at $x$ with respect to an arbitrarily high dose. In fact, $\beta(x)=R_{\mbox{\textsc{a}}}(x,a)/R_{\mbox{\textsc{a}}}(\infty,a)$. This allows inference on the BMD$_q$, which is available as the solution of the equation $\beta(x)=q/R_{\mbox{\textsc{a}}}(\infty,a)$,  for each $a \in \mathcal{Y}$  and benchmark risk $q$.

Besides the above interpretation, the formula for the conditional density $f_x(y)$ induced by model \eqref{mix}--\eqref{eq:convex1} can be intuitively seen also as the result of a process which travels in distribution between the starting location $f_{\minx}(y)$ and the ending one $f_{\maxx}(y)$, as {\em time} $x$ increases from ${\minpar} $ to ${\maxpar}$, with $\beta(x)$ denoting the proportion of the path traveled at {\em time} $x$. This metaphor eases the interpretation of $\beta(x)$, and is formally supported by Theorem 1 in the Supplementary Materials.

To conclude the discussion on model interpretation and properties note that, consistent with  \eqref{mean}, the conditional expectation $\mu(x)$ of the health outcome $y$ induced by our \comire coincides with the one characterizing popular dose--response  models discussed in Section \ref{lit} \citep{ritz:2013,ritz_2005}, after letting $\mu_{\minpar}={\boldsymbol{\nu}}_{{\minpar}}^\intercal{\bmu_{\minpar}}$, $\omega_1=1$ and $\psi_1(x)=\psi(x; {\boldsymbol{\lambda}})$.  At the same time, it is also clear that, while  relaxing the Gaussian assumption, the proposed basis expansion for $\beta(x)$ in \eqref{eq:mixIsplines} is more flexible compared to the parametric dose--response functions of routine use, thus generalizing the contributions of \citet{whee:2012} and \citet{piego_2014}.

\section{Bayesian Inference and Implementation}\label{gibbs}

Consistent with our application and the dose--response models outlined in Section~\ref{lit}, we focus on $y \in \mathcal{Y}\subseteq \R$, and let $\mbox{K}(y; \btheta) = \phi(y; \mu, \tau^{-1})$ be a Gaussian density with mean $\mu$ and variance $\sigma^2 =  \tau^{-1}$, with $\btheta = (\mu, \tau)$.  This choice also facilitates quantitative risk assessment, provided that, under the Gaussian assumption, $F_{{\minx}}(a)=\sum_{h=1}^{H} \nu_{\minpar h} \Phi\{(a-\mu_{0h}) /\sigma_{0h}\}$ and $F_{{\maxx}}(a)=\Phi\{(a-\mu_{\infty}) /\sigma_{\infty}\}$.

Before describing the Markov chain Monte Carlo  (MCMC) routine for posterior computation, we first include additional details on representations \eqref{mix}--\eqref{eq:convex2}, and  elicit the priors for each unknown quantity. Let us first introduce the augmented data $(b_i, c_i, d_i)$, for each $i=1, \ldots, n$, with $b_i \in \{1, \dots, J\}$ indicating the basis function associated with unit $i$ in \eqref{eq:mixIsplines}, $c_i \in \{1,\dots, H\}$ its mixture component at the low end of the dose range according to \eqref{eq:convex2}, and $d_i \in \{0,1\}$ an indicator for the membership to one of the two extreme mixing measures in \eqref{eq:convex1}. Conditioned on these augmented data we can define the following hierarchical relations consistent with our model:
\begin{align}
&  \left\{f_{x_i}(y) \mid d_i\right\}  = (1-d_i)f_{{\minx}}(y)+d_if_{{\maxx}}(y), \quad  d_i \sim \mbox{Bern}\{\beta(x_i)\},\nonumber \quad \\
& \left\{\beta(x_i) \mid b_i=j\right\} = \psi_{j}(x_i), \qquad b_i \sim \mbox{Cat}(w_1, \dots, w_{J}),  
 \label{data_aug}\\
&  \left\{\btheta_i \mid c_i=h,d_i=0\right\}={\btheta_{0h}}, \quad \, \, c_i \sim \mbox{Cat}(\nu_{01}, \ldots, \nu_{0H}),   \nonumber
\end{align}
where  $\mbox{Cat}(\rho_1, \dots, \rho_p)$ denotes a categorical variable with probabilities $(\rho_1, \dots, \rho_p)$,  $\mbox{Bern}(\rho)$ denotes a Bernoulli distribution with success probability $\rho$, while ${\btheta_{\minpar h}}=({\mu_{\minpar h}},{\tau_{\minpar h}})$. Marginalizing out the augmented data in \eqref{data_aug}, we obtain again   \eqref{mix},  \eqref{eq:mixIsplines} and \eqref{eq:convex2}, under the convexity assumption in \eqref{eq:convex1}.  This augmented form is useful to develop the MCMC. 

To conclude our Bayesian specification, we need to elicit prior distributions for the model parameters, including the mixture weights in \eqref{eq:convex2}, the basis coefficients of $\beta(x)$ in \eqref{eq:mixIsplines}, and each kernel-specific set of parameters. A natural prior for  ${\boldsymbol{\nu}}_{{\minpar}}=(\nu_{{\minpar}1}, \ldots, \nu_{{\minpar}H})^\intercal$ at $x=0$ is the Dirichlet distribution $\mbox{Dir}(\balpha)$ with vector of hyperparameters $\balpha = (\alpha_1, \dots, \alpha_H)$. According to  \citet{rousseau:2011}, setting small values for the elements in $\balpha$ provides a flexible prior for location-scale mixtures of Gaussians, which allows adaptive deletion of redundant components. Following a similar reasoning we define a Dirichlet prior for the coefficients of $\beta(x)$ in \eqref{eq:mixIsplines}, so that ${\boldsymbol w }=(w_1, \dots, w_{J})^\intercal\sim \mbox{Dir}(\boeta)$, with small values for the hyperparameters in $\boeta = (\eta_1, \dots, \eta_J)$. Finally, for the atoms ${\btheta_{\minpar h}}=({\mu_{\minpar h}},{\tau_{\minpar h}})$, $h=1,\ldots, H$, and \mbox{${\btheta_{\maxpar}}=({\mu_{\maxpar}},{\tau_{\maxpar}})$} characterizing the Gaussian kernels, we choose independent gamma priors $\mbox{Ga}(a_\tau,b_\tau)$ for the precision parameters, and truncated Gaussian distributions for the locations ${\mu_{\minpar h}}$, $h=1,\ldots, H$ and $\mu_{\maxpar}$, with mean $\mu \in \R$, variance $\kappa>0$, and truncations meeting the adversity profile property in \eqref{eq:order_restriction}, ${\mu_{\maxpar}}< \min_h({\mu_{\minpar h}})$.  

Under the above prior specification, samples from the posterior distribution can be obtained by sampling iteratively from the partially collapsed Gibbs sampler comprising the following steps.

\begin{enumerate}
\item Update $b_i$ from the full conditional categorical random variable having probabilities 
\[
\mbox{pr}(b_i = j \mid -) \propto w_j [ \{ 1 - \psi_j(x_i)\} f_{{\minx}}(y_i) + \psi_j(x_i) f_{{\maxx}}(y_i)],\]
 for every $j=1, \ldots, J$.
\item Update ${\boldsymbol w }$ from the full conditional Dirichlet distribution 
\[({\boldsymbol w } \mid - )  \sim \mbox{Dir}(\eta_1 + n_1, \dots, \eta_J + n_{J}),
\] where $n_j$ is the number of subjects in which $b_i=j$, and update $\beta(x)$ by applying equation \eqref{eq:mixIsplines}.
\item Update $d_i$ from 
\[
(d_i \mid - )\sim \mbox{Bern}\left[  \frac{\beta(x_i) f_{\maxx}(y_i)}{\{1-\beta(x_i) \} f_{\minx}(y_i) + \beta(x_i) f_{\maxx}(y_i)}\right].
\]
\item For individuals with $d_i=0$, update $c_i$  from the categorical variable having probabilities equal to  
\[
\mbox{pr}(c_i = h \mid -) \propto \nu_{{\minpar}h} \phi(y_i; {\mu_{\minpar h}}, {\tau^{-1}_{\minpar h}}),
\] for every $h=1, \ldots, H$.
\item Update the mixture weights ${\boldsymbol{\nu}}_{{\minpar}}$ characterizing $P_{\minx}$ from 
\[
({\boldsymbol{\nu}}_{{\minpar}} \mid - ) \sim \mbox{Dir}\left(\alpha_1 + n_{1}, \dots, \alpha_H + n_{H} \right),\]
 where $n_{h}$ is the number of units with $d_i=0$, having $c_i=h$.
\item  Update each ${\btheta_{\minpar h}}=({\mu_{\minpar h}},{\tau_{\minpar h}})$, $h=1, \ldots, H$ under the restriction ${\mu_{\maxpar}}< \min_h({\mu_{\minpar h}})$, from 	\begin{eqnarray*}
&&({\mu_{\minpar h}}\mid -) \sim \mbox{N}(\hat{\mu}_h,\hat{\sigma}^2_h, \mu_{\maxpar}, +\infty), \\ 
&&({\tau_{\minpar h}} \mid -)\sim \text{Ga}(\hat a_{h}, \hat b_{h}),
\end{eqnarray*}
 where $\mbox{N}( \mu, \sigma^2,a,b)$ denotes a Gaussian  with mean $\mu$, variance $\sigma^2$, and truncated to the interval $(a,b)$. In the above equation, $\hat a_{h} = a_\tau + n_h/2$, $\hat b_{h} = b_\tau + \sum_{i=1}^n\I(d_i=0,c_i=h) (y_i - {\mu_{\minpar h}})^2/2$, $\hat{\sigma}^2_h=(\kappa^{-1}+n_h{\tau_{\minpar h}})^{-1}$, $\hat \mu_h = \hat{\sigma}^2_h(\kappa^{-1} \mu + n_h {\tau_{\minpar h}}\bar{y}_h)$, and $\bar{y}_h$ is the mean of $y$ for the units with  $d_i=0$ and $c_i=h$. The term $\I(\cdot)$ denotes, instead, the indicator function.
\item Finally, update  ${\btheta_{\maxpar}}=({\mu_{\maxpar}},{\tau_{\maxpar}})$ under the  restriction ${\mu_{\maxpar}}< \min_h({\mu_{\minpar h}})$,  from 		
	\begin{eqnarray*}
&& ({\mu_{\maxpar}}\mid -)\sim \mbox{N}(\hat{\mu}_{\maxpar},\hat{\sigma}^2_{\maxpar}, -\infty,  \mbox{min}_h({\mu_{\minpar h}})), \\
&& ({\tau_{\maxpar}} \mid -)\sim \text{Ga}(\hat a_{\maxpar}, \hat b_{\maxpar}), 
\end{eqnarray*}
where $\hat a_{\maxpar}= a_\tau + n_{\infty}/2$, $\hat b_{\maxpar} = b_\tau + \sum_{i=1}^nd_i(y_i - {\mu_{\maxpar}})^2/2$, $\hat{\sigma}^2_{\maxpar}=(\kappa^{-1}+n_{\maxpar}{\tau_{\maxpar}})^{-1}$, $\hat \mu_{\maxpar} = \hat{\sigma}^2_{\maxpar}(\kappa^{-1} \mu + n_{\maxpar} {\tau_{\maxpar}}\bar{y}_{\maxpar})$, $\bar{y}_{\maxpar}=\sum_{i=1}^nd_iy_i/n_{\maxpar}$, and $n_{\maxpar}=\sum_{i=1}^nd_i$.
\end{enumerate}
This routine can be easily generalized to incorporate mixture models also for $f_{\maxpar}(y)$ simply modifying step (7), and adding a step similar to (4), to update the  weights also at the high end of $\mathcal{X}$. 

\subsection{Goodness-of-fit via posterior predictive checks}
\label{check}
To assess whether the restrictions of \comire provide a good fit to the observed data under the Bayesian framework, we perform posterior predictive checks \citep{gelman2013} comparing relevant summary statistics computed from the observed data with their corresponding  posterior predictive densities. Under  \comirep, the posterior predictive density  $f_x(y) \mid (y_1,x_1), \ldots, (y_n,x_n)$ from which the response data can be simulated is defined as
\begin{eqnarray}
\begin{split}
\int  \Bigg[ & \{1 - \beta(x)\}  \sum_{h=1}^H \nu_{0h} \tau_{0h}\phi\{\tau_{0h}(y - \mu_{0h})\} +  \\    &  \beta(x) \tau_{\maxpar}\phi\{ \tau_{\maxpar}(y - \mu_{\maxpar})\} \Bigg] d \Pi({\boldsymbol{\nu}}_{{\minpar}}, {\boldsymbol{\theta}}_{{\minpar}},{\boldsymbol{\theta}}_{{\maxpar}},{\boldsymbol{w}}  \mid  \mbox{data}),
\end{split}
\label{post_pred}
\end{eqnarray}
where $\Pi({\boldsymbol{\nu}}_{{\minpar}}, {\boldsymbol{\theta}}_{{\minpar}},{\boldsymbol{\theta}}_{{\maxpar}},{\boldsymbol{w}} \mid \mbox{data})
$ denotes the joint posterior distribution of the model parameters. 

Although \eqref{post_pred} is not analytically available, it is straightforward to simulate outcome data $y$ from the posterior predictive distribution in \eqref{post_pred} by leveraging the MCMC samples of the  parameters and the hierarchical representation of the model in \eqref{data_aug}. Once these simulated data sets are available, different summary statistics of interest can be computed, thus obtaining samples from the corresponding  posterior predictive distribution. If \comire is not sufficiently flexible for a specific study, we expect the observed   summary statistics to fall in the tails of their associated posterior predictive density. Since our main interest is in quantitative risk assessment, we focus on comparing a smoothed empirical estimate of $\mbox{pr}(y\leq a \mid x)=F_x(a)$ computed from the dichotomized sample data $\{\I(y_1\leq a),x_1\}, \ldots, \{\I(y_n\leq a),x_n\}$, with its posterior predictive density obtained from the replicated data sets generated from \eqref{post_pred}. Although other statistics can be considered, according to \eqref{add_risk}, $F_x(a)$ is the key function underlying $R_{\mbox{\textsc{a}}}(x,a)$, thus requiring careful model checking.

\section{Simulation Experiments}
\label{simu}

To study the empirical performance of \comire in providing inference for $R_{\mbox{\textsc{a}}}(x,a)$ under several data generative processes, we consider three simulations, covering correctly specified models and model misspecification scenarios. To carefully assess performance in reasonable applied settings, these scenarios are defined to incorporate relevant aspects of the motivating CPP application and other related observational studies. These assessments involve two different sample sizes: $n=2000$ to mimic the CPP application, and $n=500$ to evaluate performance in smaller studies.

A common feature in  observational dose--response studies is that the data are increasingly sparse as the dose grows, with most of the individuals having reasonably low exposures to adverse chemicals. We incorporate this feature by generating the doses $x_i$, $i=1, \ldots, n$, from a gamma distribution with a heavy right tail. These simulated doses range from $0$ to about $132$, and have empirical quartiles $15.4$, $24.9$,  $39.8$ for the scenarios with $n=2000$, and $11.2$, $25.0$, $48.2$, for those with $n=500$. For each unit,  $y_i$ is generated under three different mechanisms which maintain a similar range and structure to the one characterizing the observed gestational age in the CPP data.

 \begin{figure*}
  \centering
     \subfigure[Scenario 1]{\includegraphics[width=.32\textwidth]{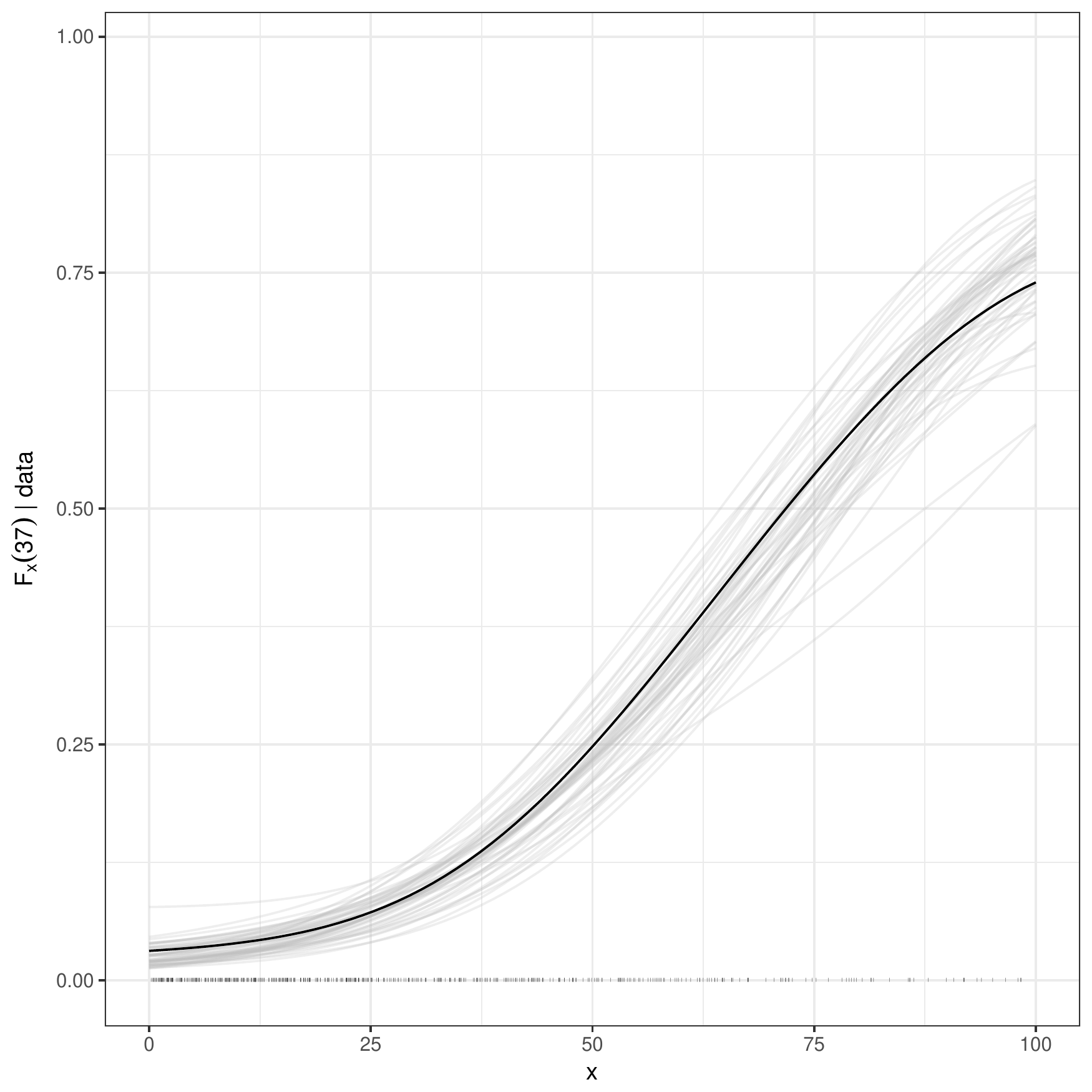}}
          \subfigure[Scenario 2]{\includegraphics[width=0.32\textwidth]{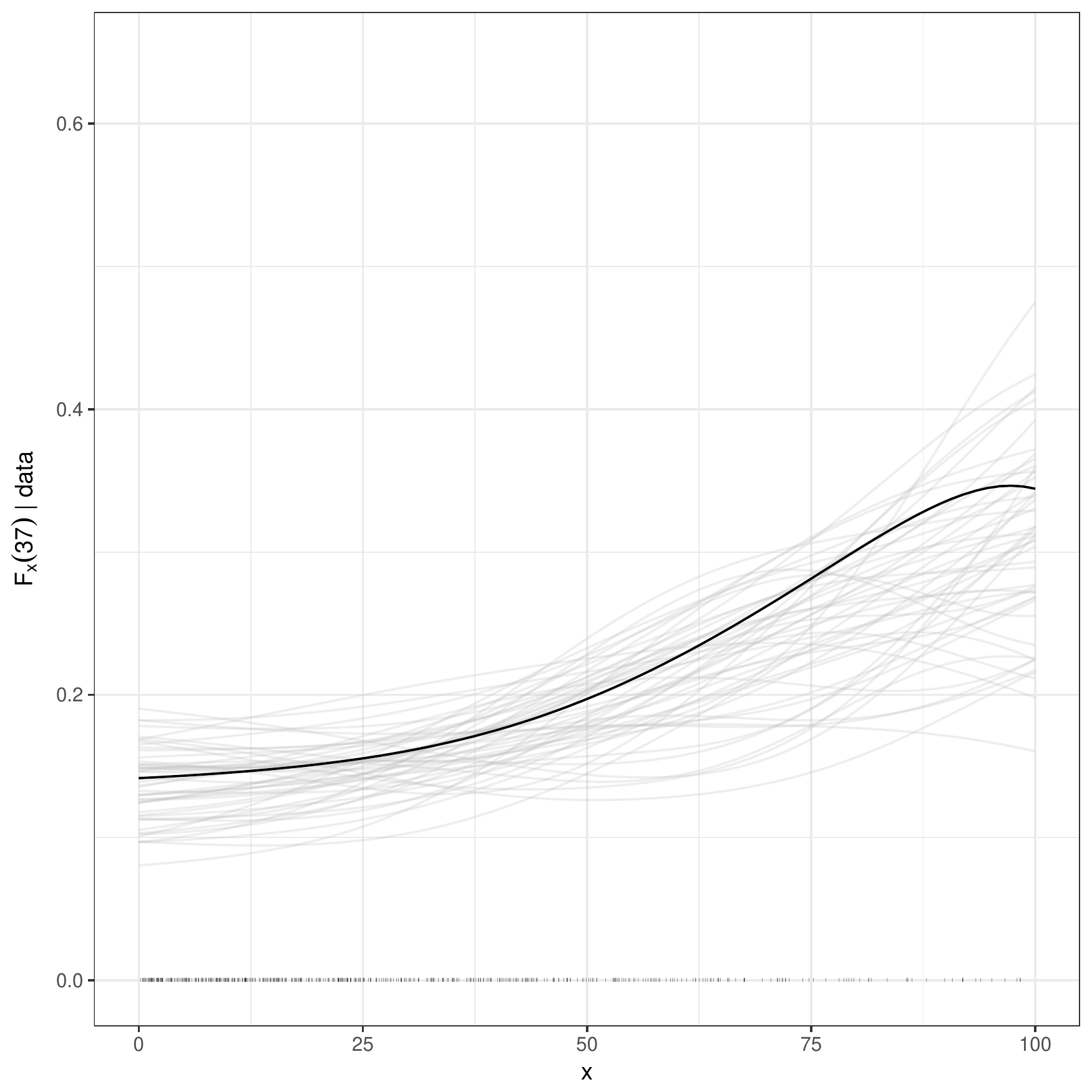}}
               \subfigure[Scenario 3]{\includegraphics[width=0.32\textwidth]{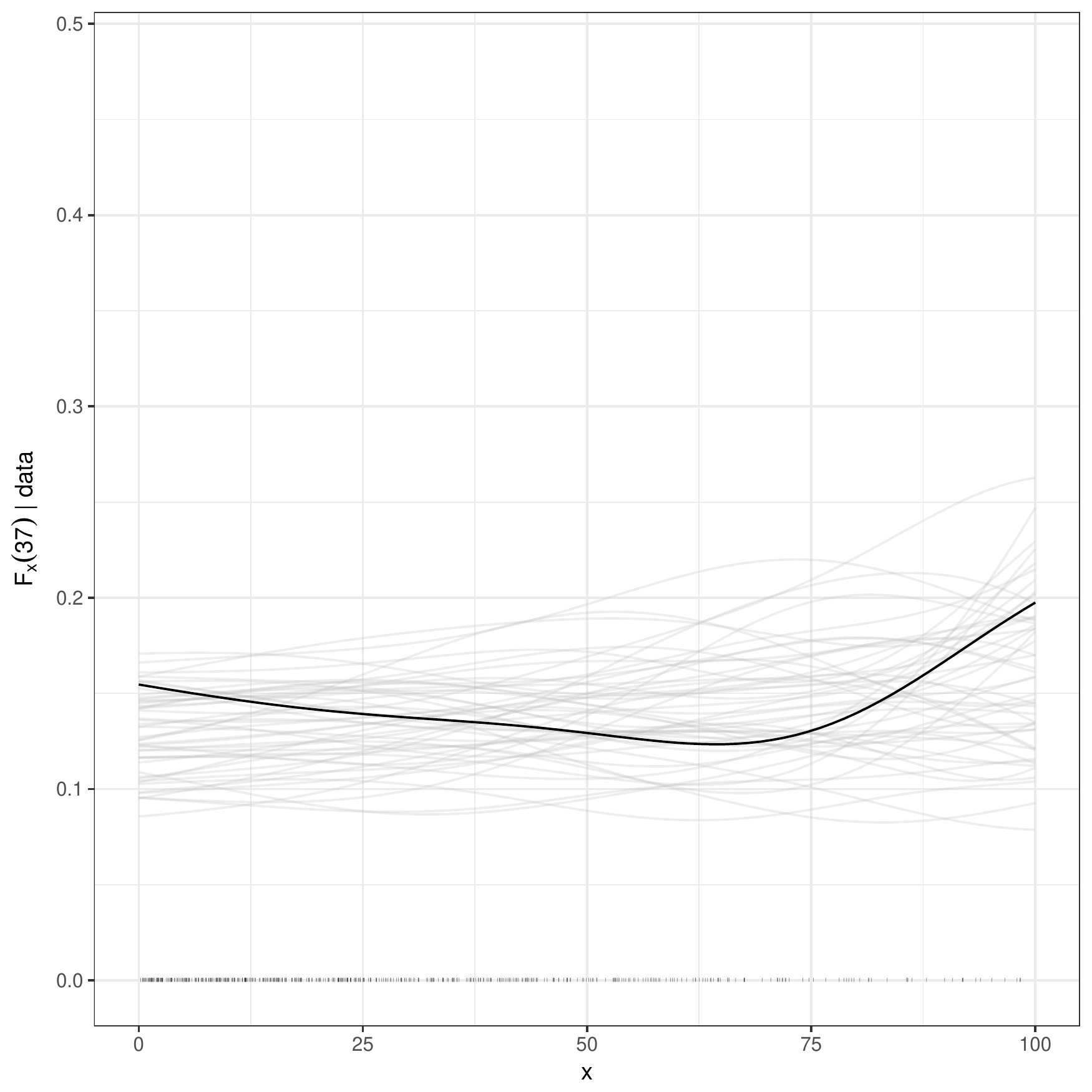}}
   \caption{Goodness-of-fit assessments for $n=500$ in Scenario 1 (a), 2 (b), and 3  (c). Smoothed empirical estimate of $F_x(37)=\mbox{pr}(y \leq 37 \mid x)$ computed from the observed data (black line), and from $50$ data sets simulated from the posterior predictive distribution induced by  \comire (grey lines). In the $x$ axis we also report the simulated dose exposures.}
   \label{fig:postpredcheck}
\end{figure*}

Scenario 1 covers the correctly specified setting and assumes the proposed \comire model where the extreme density at $x=0$ is a mixture of three Gaussian kernels having ${\bmu_{\minpar}}=(37, 39, 40)^\intercal$, ${\boldsymbol{\nu}}_{{\minpar}}=(0.05, 0.15, 0.80)^\intercal$, and unit variances, whereas the density at $x\rightarrow \infty$ is a single Gaussian kernel with ${\mu_{\maxpar}}=36$ and ${\sigma^2_{\maxpar}}=1$.  The $\beta(x)$ function is assumed to be the cumulative distribution function of a gamma random variable with shape 6 and rate 0.1.  Focusing on performance under model misspecification, we generate the data from a dependent process but without \comire specification in Scenario 2, and under an independent process in Scenario~3. Specifically, Scenario~2 assumes a  mixture of Gaussians with unit variances and locations changing with $x$. The mixing weights are $(0.10, 0.25, 0.65)$, whereas the locations are $\mu_1(x) = -x/300 + 35.5$, $\mu_2(x) =-x/50 + 38.5$, and $\mu_3(x) = -x/75 +40.5$, respectively. Scenario 3 is instead a mixture of Gaussians with weights equal to $(0.25, 0.25, 0.50)$, locations (37, 39, 41), and unit variances.

In performing posterior inference, we set the hyperparameters following the guidelines in Section~\ref{gibbs}. In particular we fix $H=10$ for the number of mixture components at zero dose, and set small values  $\alpha_1= \dots= \alpha_H=1/H$ in the Dirichlet prior for ${\boldsymbol{\nu}}_{{\minpar}} $ to allow adaptive shrinkage. The function $\beta(x)$ in \eqref{eq:mixIsplines} is instead defined via cubic I-splines with $7$ equally-spaced inner knots and $2$ more at the extremes of the observed dose range---i.e. 0 and 132. The  knot at $x=0$ incorporates the restriction $\beta(0)=0$, whereas the one at 132 introduces a zero constant function which allows asymptotes in $\beta(x)$ at extreme doses. This choice provides $J=10$ bases. By considering small values $\eta_1=\cdots=\eta_J=1/J$ for the hyperparameters in the Dirichlet prior for the I-splines coefficients ${\boldsymbol w }$, we facilitate automatic learning of which bases are required to flexibly characterize the dose--response function $\beta(x)$. Finally, we center the mean parameters in the Gaussian kernels on the empirical mean of the observed responses $\mu = \bar{y}$, and allow moderate variations in the expectation of the different components by setting $\kappa = 10$. For the hyperparameters  in the component--specific variances we set instead $a_\tau=b_\tau = 2$ to allow a moderate variability within each Gaussian kernel. To assess sensitivity we maintain the same default hyperparameters settings in all the simulations.

Posterior inference relies on $5000$ Gibbs iterations after a burn-in of $2000$, and thinning the chains every 5 samples. These settings were sufficient to obtain convergence and good mixing---based on the traceplots and the Geweke's diagnostic. The computations for each simulation took $\approx 2'$ and $\approx 6'$ for $n=500$ and $n=2000$, respectively,  under an \texttt{R} implementation  on a MacBook Pro with 2.8 GHz Intel i7 processor and 16 GB RAM. The following discussions are related to the case $n=500$. The  results for $n=2000$ are similar and reported in the Supplementary Materials. It shall be also noticed that the following figures, including those in Section \ref{cpp} and  in the Supplementary Materials, evaluate performance in a large dose range avoiding excessively high exposures, which are unlikely to occur in the population and are not of direct interest in quantitative risk assessments.

We study performance with a particular focus on recovering the true additional risk $R_{\mbox{\textsc{a}}}(x,a)$.  However, before focusing on $R_{\mbox{\textsc{a}}}(x,a)$, we first check model adequacy in Figure \ref{fig:postpredcheck}, following the guidelines provided in Section \ref{check}.  As is clear from Figure \ref{fig:postpredcheck}, our approach  provides adequate fit to the smoothed empirical estimate of $F_x(37)=\mbox{pr}(y \leq 37 \mid x)$ computed from the observed data, thereby motivating inference on the additional risk $R_{\mbox{\textsc{a}}}(x,a)$, with $a=37$.  Note that these plots can detect lack of fit---for example, if the dose--response relation is non--monotone, or if there are evident differences between the generative process of the observed data and the deflation--inflation mechanism of \comirep. In these situations, relaxing \comire restrictions under more flexible models is important to avoid bias. We shall however emphasize that non--monotonicity of dose--response across the range of exposures faced by humans in usual settings is rare. The most common type of non--monotonicity in these studies corresponds to down-turns at high doses---typically well beyond normal ranges of human exposures.  Similarly, evident departures from the \comire assumptions, although possible, are not frequent in observational dose--response studies.  Hence, we focus on monotone relations, and evaluate robustness to misspecification in reasonable scenarios.

\begin{figure*}
  \centering
   \subfigure[Scenario 1]{\includegraphics[width=0.32\textwidth]{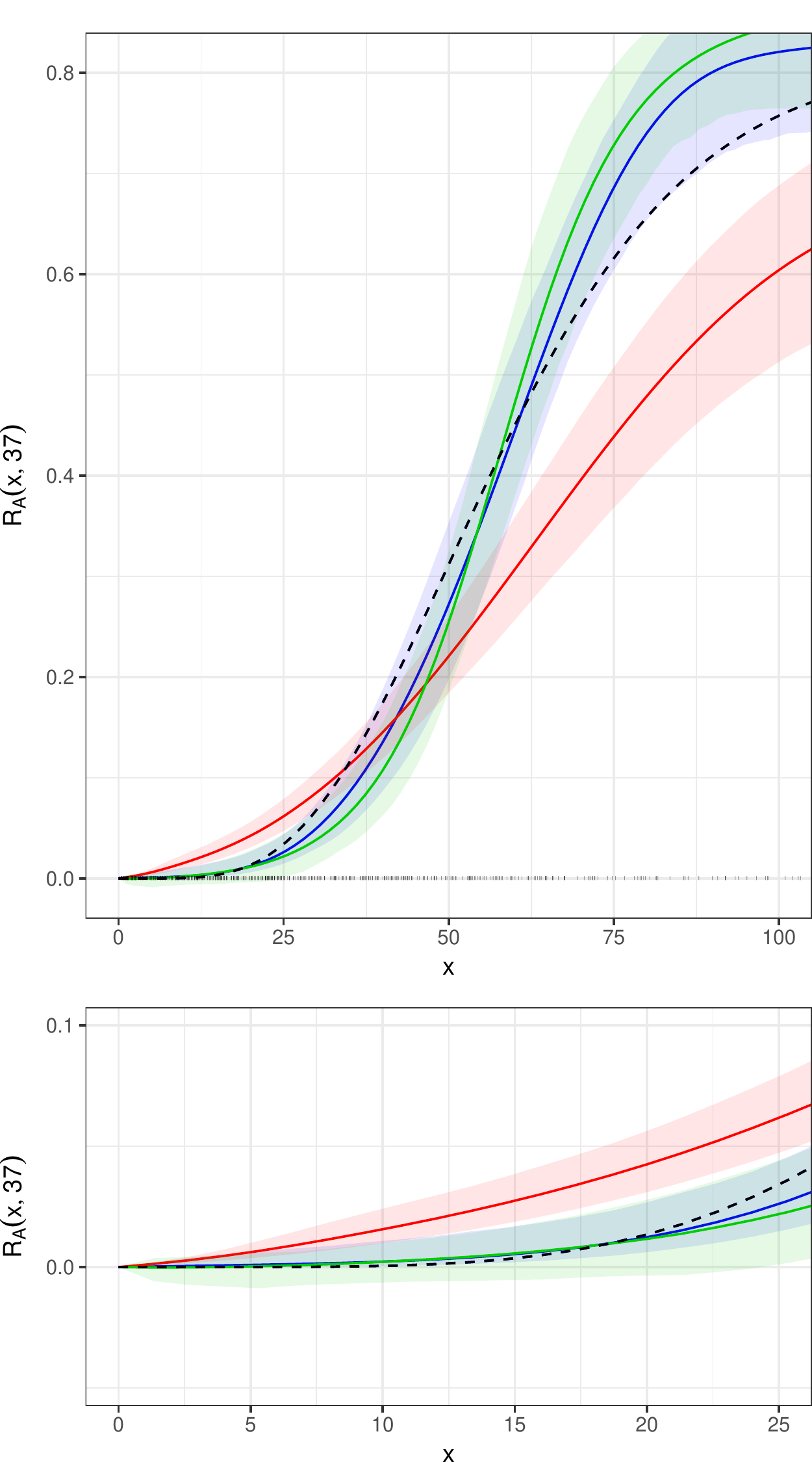}}
      \subfigure[Scenario 2]{\includegraphics[width=0.32\textwidth]{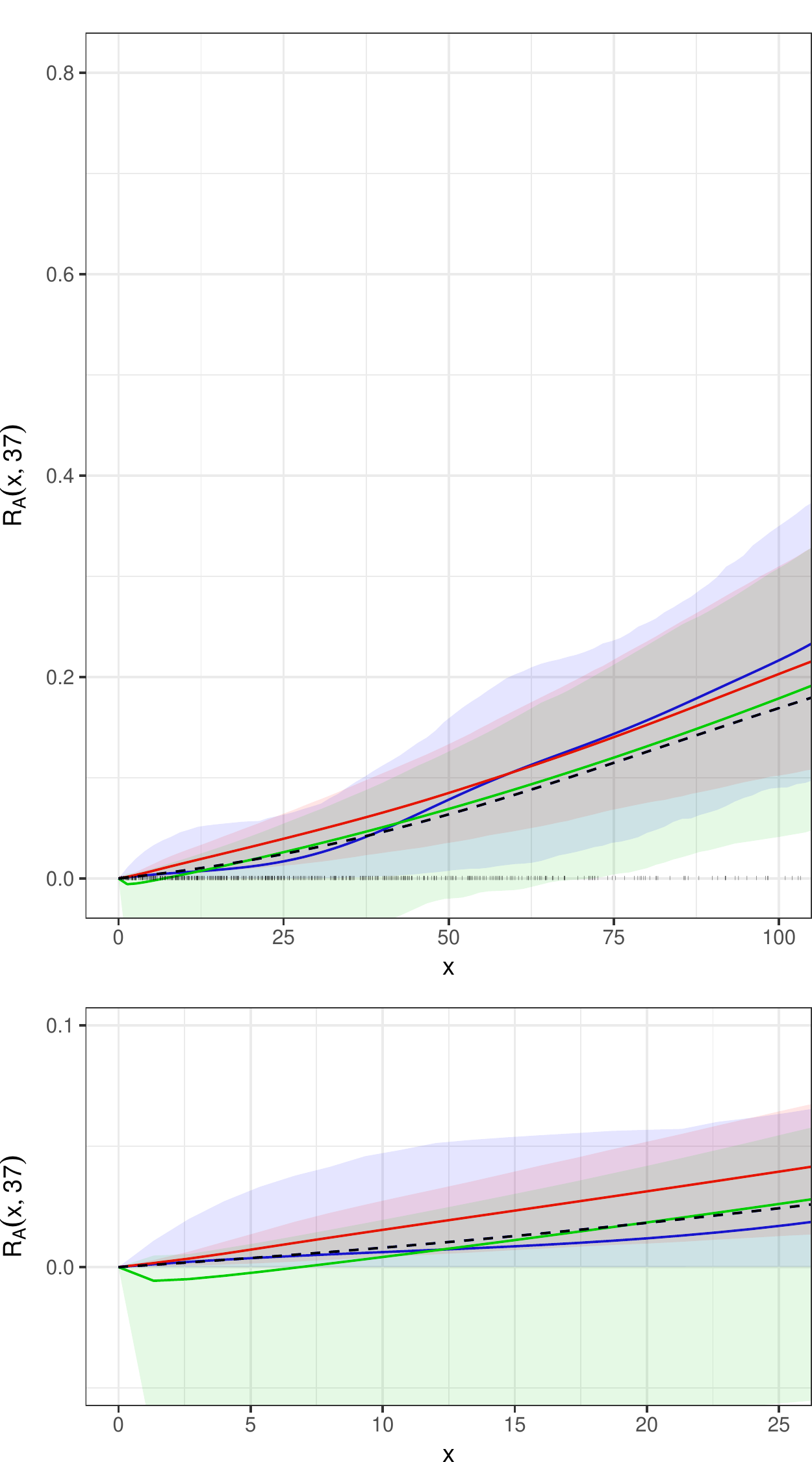}}
         \subfigure[Scenario 3]{\includegraphics[width=0.32\textwidth]{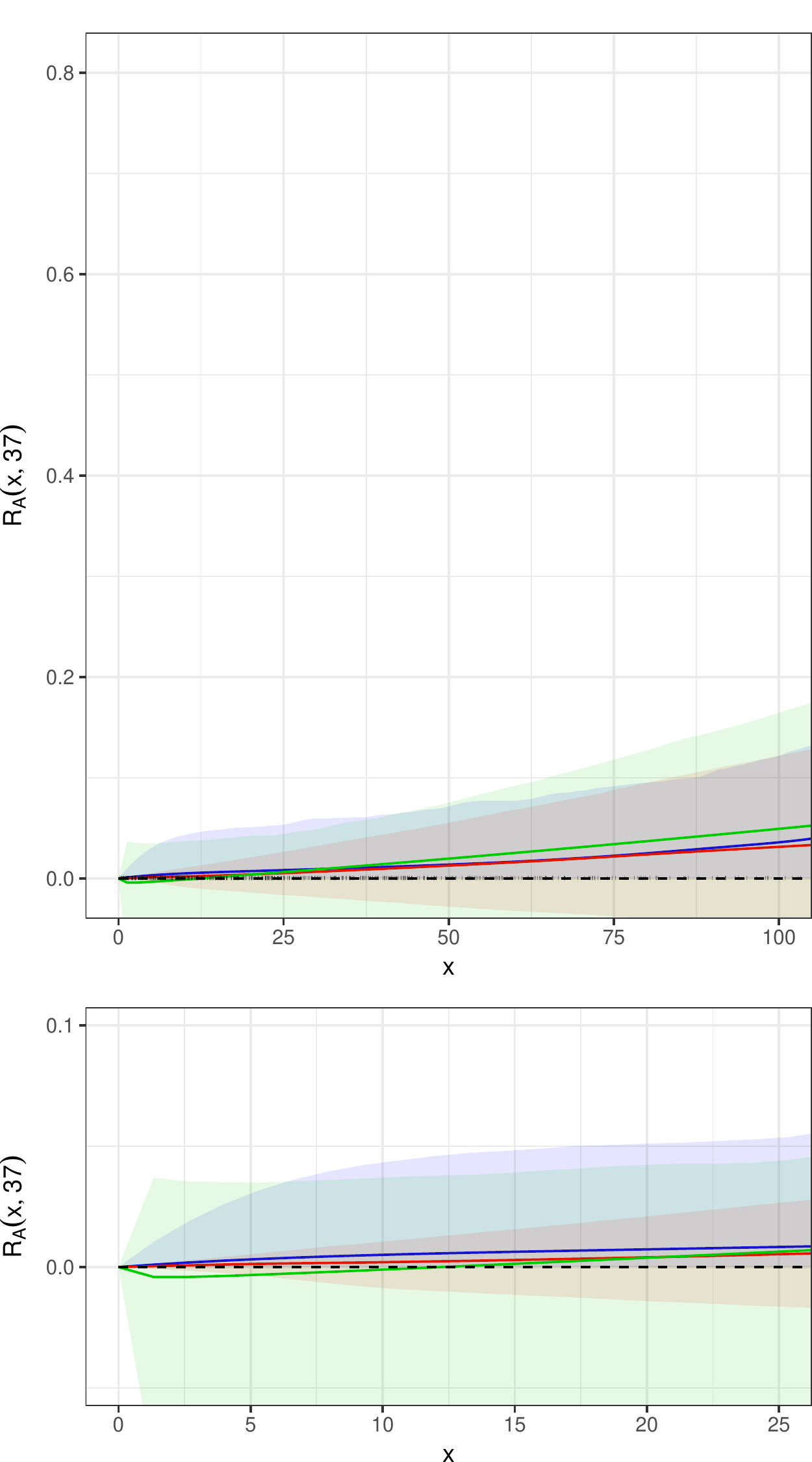}}
       \caption{Inference on the additional risk function for $n=500$ in the three scenarios.  The dashed lines represent the true additional risk function  $R_{\mbox{\textsc{a}}}(x,37)$, whereas the red, green, and blue continuous lines denote the posterior mean of $R_{\mbox{\textsc{a}}}(x,37)$ under \textsc{anova--ddp}, \textsc{f--dmix}, and \comirep, respectively. The shaded areas represent the pointwise 95\% posterior credible bands. In the $x$ axis we report the simulated dose exposures. Lower panels provide a zoom on the range of the additional risk typically considered in benchmark dose analysis. }
   \label{fig:sim_inference1}
\end{figure*}

To highlight the benefits of \comire in providing inference on the additional risk, we compare performance with the relevant competitors discussed in Section \ref{lit}. In particular, we consider a classical Gaussian regression  with a Weibull dose--response function for $\mu(x)$  \citep{ritz:2013,ritz_2005}, a more general formulation of the model in \citet{razzaghi2000}, via \textsc{anova--ddp} with fixed weights \citep{deio:etal}, and, finally, a fully flexible dependent mixture of Gaussians (\textsc{f--dmix}) including changes also in the mixing probabilities as in \citet{he:etal2010}. These models are estimated with the \texttt{R} packages \texttt{drc}, \texttt{DPpackage}, and \texttt{LSBP}, respectively. \texttt{LSBP} is a recent implementation of the probit stick--breaking by \citet{rodriguez:2011}, using the logistic link to improve computational performance without affecting flexibility \citep{rigon:2017}. In performing posterior inference under the \textsc{anova--ddp} and the \textsc{f--dmix}, we consider routine implementations via linear functions of the predictor, and set the hyperparameters using default choices which induce a comparable prior uncertainty to  \comirep. Figure \ref{fig:sim_inference1} summarizes  the posterior mean function and  pointwise 95\% posterior credible bands for $R_{\mbox{\textsc{a}}}(x,37)$ obtained under \comirep,  \textsc{anova--ddp} and \textsc{f--dmix}, in the three scenarios. Results under the  Weibull dose--response function were substantially inferior, and therefore are not reported here.

The results in Figure~\ref{fig:sim_inference1} confirm our discussion in Section \ref{lit} on the different approaches to inference in quantitative risk assessment. In particular, \comire allows improved learning of  $R_{\mbox{\textsc{a}}}(x,37)$ in the correctly specified scenario, and has comparable performance to the more flexible \textsc{anova--ddp} and \textsc{f--dmix}, in misspecified settings. As it can be noticed in the lower panels of Figure~\ref{fig:sim_inference1}, the estimation of $R_{\mbox{\textsc{a}}}(x,37)$ is particularly precise at low--dose exposures $x$, where quantitative risk assessments typically focus. The reduced performance at high--dose exposures is mainly due to sparsity in the data, thereby providing reduced information to effectively estimate the parameters characterizing $P_{\maxpar}$. This result is also evident in  Figure~\ref{fig:postpredcheck}. The \textsc{anova--ddp} provides accurate inference on $R_{\mbox{\textsc{a}}}(x,37)$  in Scenarios 2 and 3, where this model is correctly specified, but induces a notable bias in Scenario 1, possibly due to the assumption of constant mixing weights. When relaxing this restriction under  \textsc{f--dmix}, the bias is reduced in Scenario 1. However, consistent with the discussion in Section \ref{lit}, the increased flexibility associated with \mbox{\textsc{f--dmix}} requires more parameters, thereby reducing efficiency. This is evident in all the three scenarios, with  \textsc{f--dmix} having higher posterior uncertainty than \comirep, without substantially improving performance in estimating $R_{\mbox{\textsc{a}}}(x,37)$. According to the lower panels in Figure \ref{fig:sim_inference1}, this reduced efficiency is evident for the common range $(0,0.1]$ of the additional risk to be considered in benchmark dose analysis. These conclusions did not change when considering other thresholds $a$ different from $37$.

Similar results also are discussed for $n=2000$ in the Supplementary Materials. In this case the higher sample size further reduces bias and posterior uncertainty---as expected.

\section{Analysis of the CPP Data}
\label{cpp}
We conclude by applying \comire to infer changes in the gestational weeks at delivery with \textsc{dde}, as discussed in Section \ref{intro}. In particular, our focus is on $n=2312$ women with gestational ages less or equal than 45 weeks, since higher values are clear measurement errors.  Posterior inference is performed with the default hyperparameters and MCMC settings carefully described in the tutorial implementation at \url{github.com/tonycanale/CoMiRe}. Also in this case we obtain convergence and good mixing, with a runtime \mbox{of $\approx 8'$.}

As outlined in Figures~\ref{fig:motivate} and \ref{fig:cpp_pcc},  \comire provides an adequate fit to the observed data and the associated functionals of interest, thus motivating inference and quantitative risk assessments.  According to Figure~\ref{fig:motivate}, the conditional density of the  gestational age at delivery is far from being Gaussian and displays variability, skewness and multimodality patterns changing with \textsc{dde}.  Indeed, at low--dose exposures most of the probability mass is concentrated around normal pregnancies, with the posterior mean and the 95\% credible interval for $\mu(0)$ being $40.20$ and $(40.01,40.34)$. This is in line with results on normal gestational ages measured via last menstrual periods, as in our application \citep{longnecker2001}. In addition to these findings, we also notice mild negative skewness, with a minor mode closer to overdue pregnancies, and a tail towards preterm deliveries. As \textsc{dde} exposure grows, the negative skewness is still maintained, but the tail towards preterm deliveries increasingly inflates until characterizing a more evident mode centered on $\mu_{\maxpar}$. The posterior mean and the 95\% credible interval for the latter quantity are $36.22$ and $(35.36,36.94)$, respectively, thus providing evidence of preterm pregnancy profiles at high--dose exposures. 

\begin{figure}
   \centering
   \includegraphics[width=0.47\textwidth]{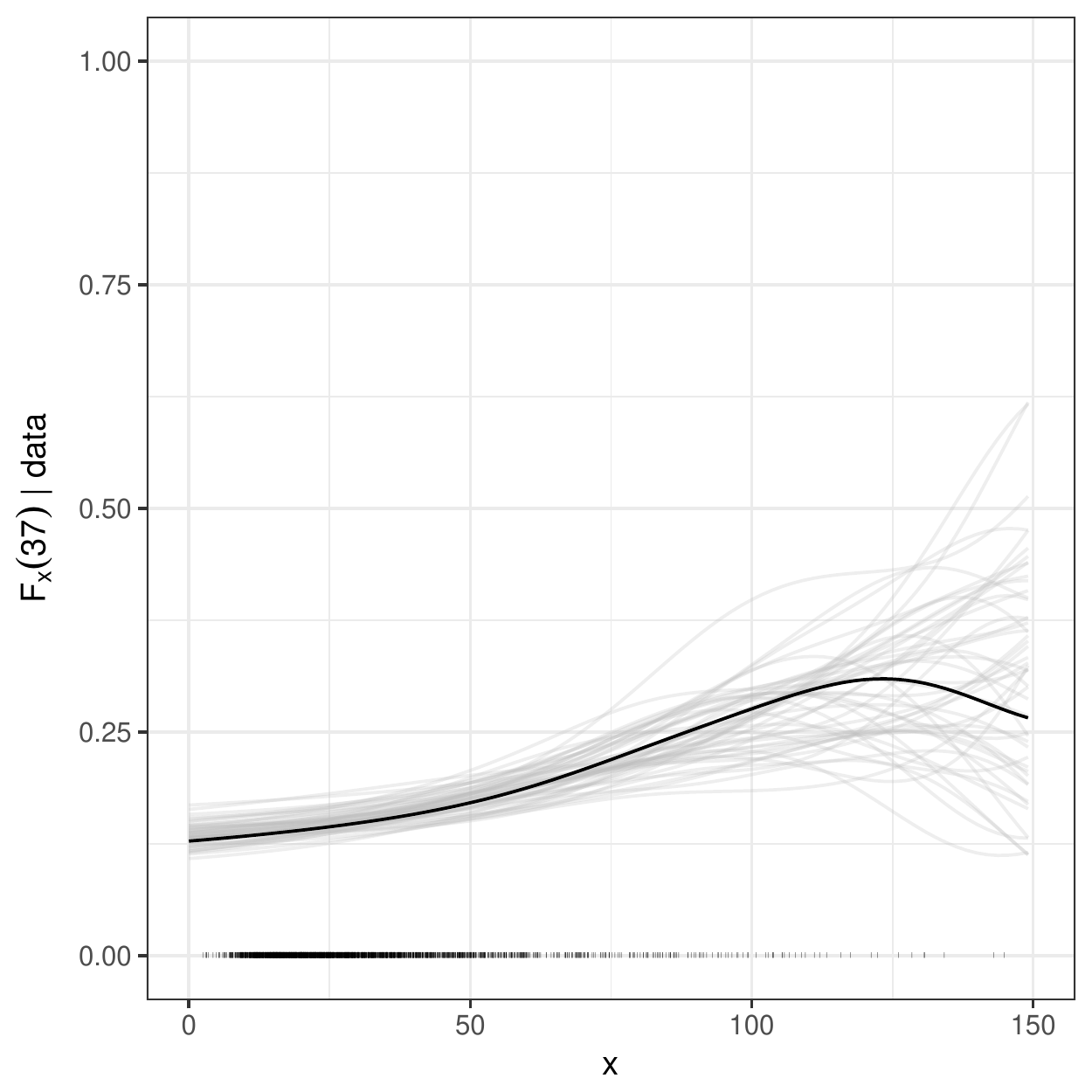}
   \caption{Goodness-of-fit assessment in the application. Smoothed empirical estimate of $F_x(37)=\mbox{pr}(y \leq 37 \mid x)$ computed from the observed data (black line), and from $50$ data sets simulated from the posterior predictive distribution induced by  \comire (grey lines). In the $x$ axis we report the observed exposures.}
   \label{fig:cpp_pcc}
\end{figure}

\begin{figure} 
   \centering
            \subfigure[]{\includegraphics[width=0.32\textwidth]{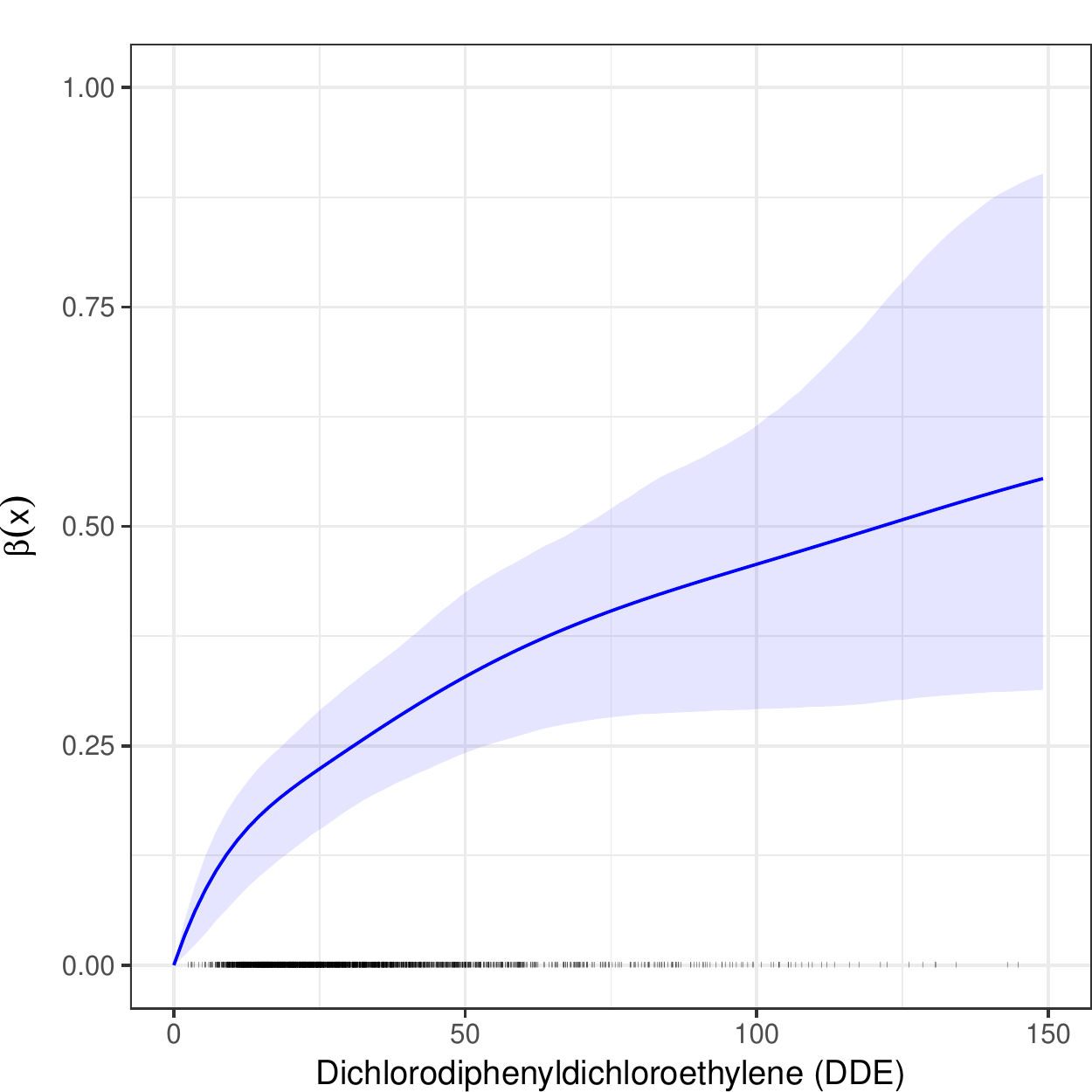} }
            \subfigure[]{\includegraphics[width=0.32\textwidth]{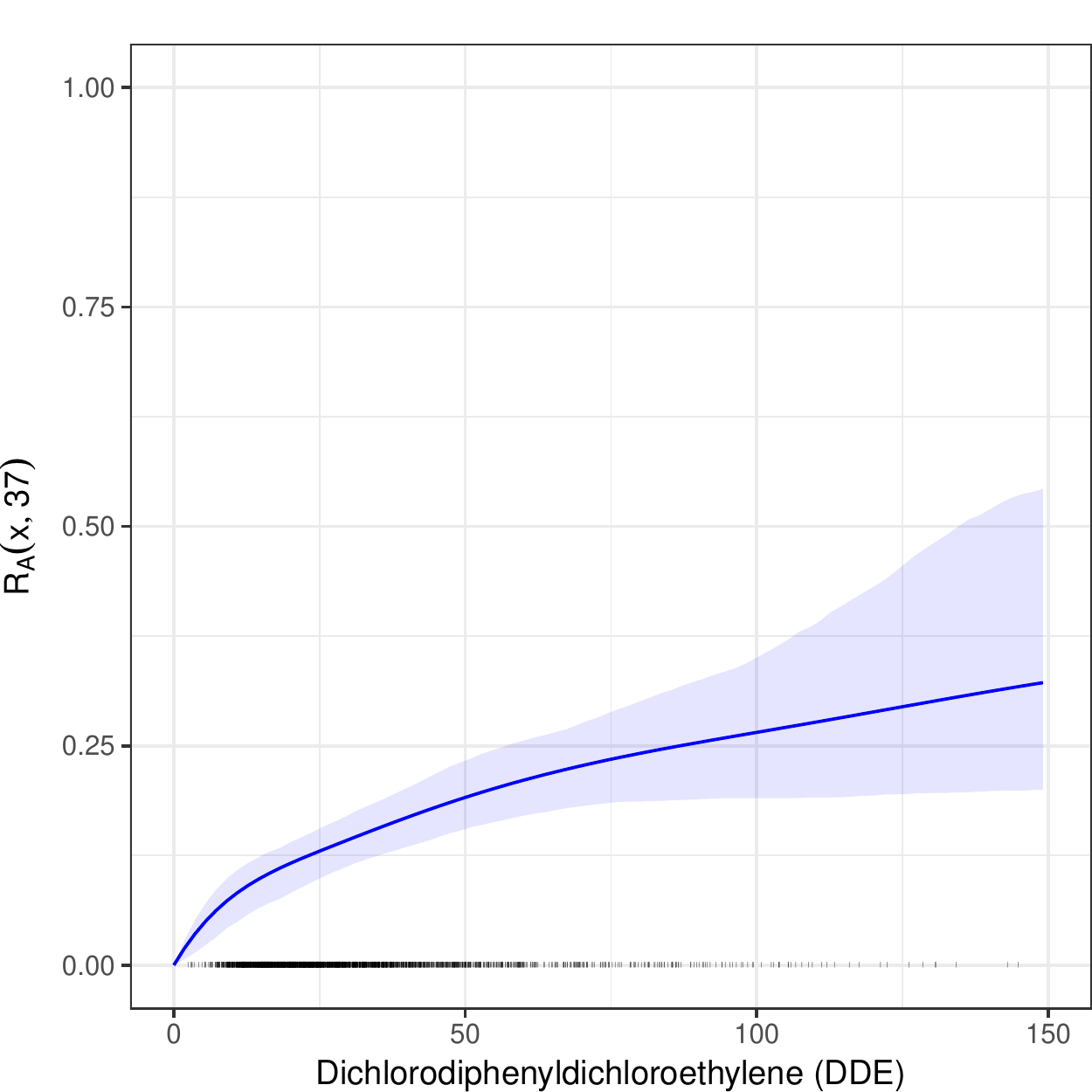} }
            \subfigure[]{\includegraphics[width=0.32\textwidth]{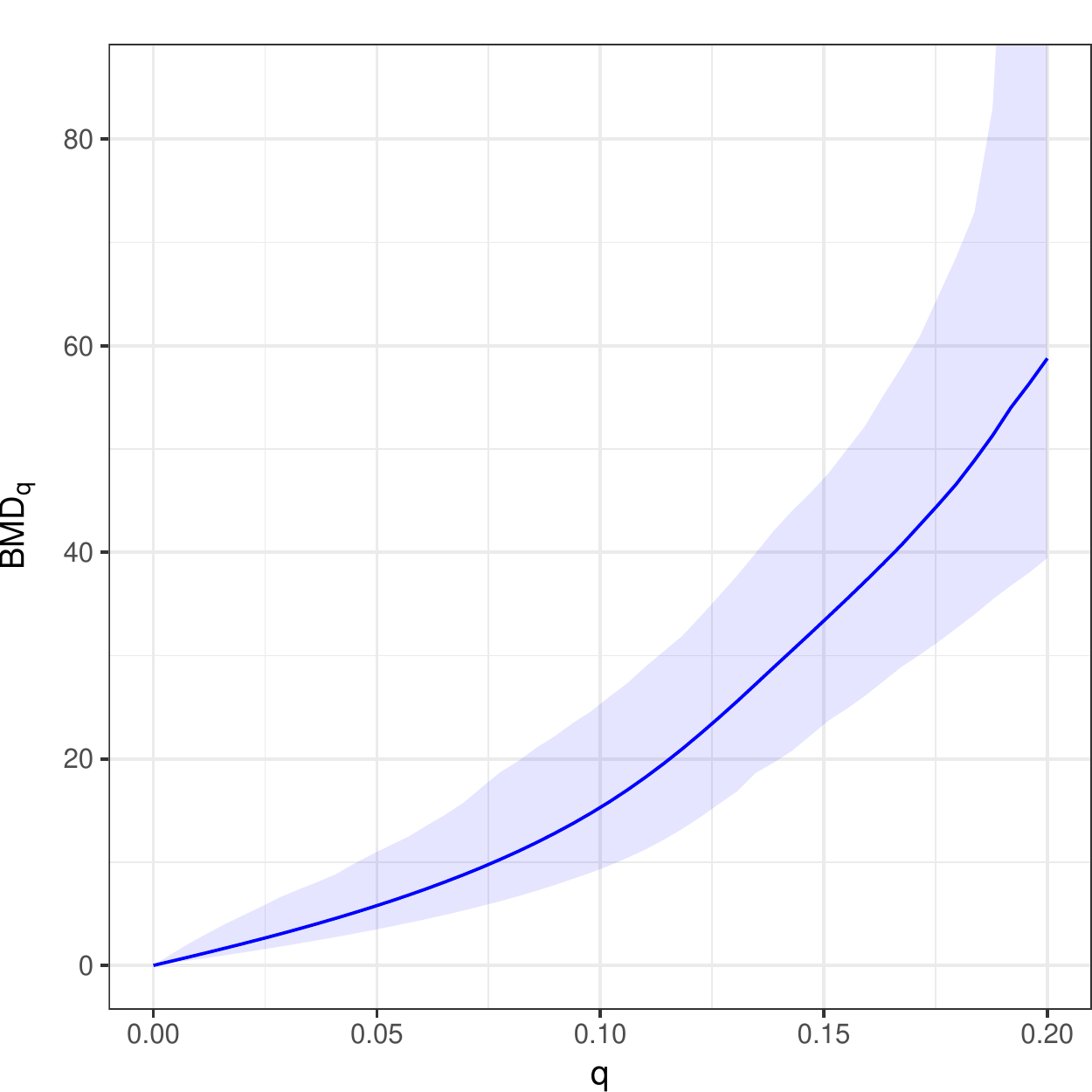} }
   \caption{Posterior mean (solid lines) and pointwise 95\% credible bands (shaded areas) for (a) $\beta(x)$, (b) $R_{\mbox{\textsc{a}}}(x,37)$ and (c) the related BMD$_q$. In the $x$ axis in (a) and (b) we report the observed exposures.  }
   \label{fig:inference_cpp}
\end{figure}

These results clearly show that the negative association of  \textsc{dde}  with gestational age are mostly on the shape, rather than on the trend. In particular, recalling the above discussion, both adverse and non-adverse profiles are found across the predictor space, including at no dose and high--dose exposures. What changes with  \textsc{dde}  is the degree $\beta(x)$ of susceptibility to the adverse effects of this chemical. In particular, the posterior mean and the 95\% credible bands of the $\beta(x)$ function in panel (a) of Figure \ref{fig:inference_cpp}, show a notable increment in the probability of the most adverse health profile at low--dose exposures. For instance, about $20\%$  of the women with a low  \textsc{dde}  concentration $\approx 20 \mu g/l$ are expected to have a gestational age comparable to women who had high  \textsc{dde}  exposures, with this profile being indicative of preterm pregnancies. Consistent with the properties of \comire discussed in Section \ref{prop}, this percentage explicitly measures in which proportion the additional risk at a given   \textsc{dde}  exposure relates to the one associated with the worst health profile observed at an arbitrarily high dose without the need to specify any threshold $a$ for the negative health event.

For benchmark dose analyses, we consider the standard preterm threshold \mbox{$a=37$.} Consistent with Section  \ref{prop}, rescaling $\beta(x)$ by $R_{\mbox{\textsc{a}}}(\infty,37)$ provides the additional  risk function $R_{\mbox{\textsc{a}}}(x,37)$---reported in panel (b) of Figure \ref{fig:inference_cpp}---inheriting the notable increment in the additional risk at low--dose exposures, which suggests conservative benchmark doses.  The posterior mean and the $95\%$ credible bands for the BMD$_q$---expressed as a function of $q$---confirm these findings. In particular, as shown in panel (c) of  Figure~\ref{fig:inference_cpp}, the BMD$_q$ evolves on low values, especially for the range of benchmark risks $q \in \{0.01, 0.05, 0.10\}$ of interest in toxicology \citep[e.g.][]{piego:2005}, thus implying the conservative policies on \textsc{dde} exposures  in Table~\ref{tab:bmd}. Motivated by interest in more conservative benchmark doses BMDL$_q$  relying on a lower confidence bound of the BMD$_q$, Table~\ref{tab:bmd} provides also the BMDL$_q$. Under our Bayesian approach to inference, such exposures can be obtained as the $5\%$ lower quantile of the posterior distribution of the BMD$_q$, instead of requiring frequentist asymptotic approximations.

\begin{table}
\centering
\begin{tabular}{lrrr}
  \hline
  & {\bf 0.01} & {\bf 0.05} & {\bf 0.10}  \\ 
  \hline
 BMD$_q$ & 1.03 {\scriptsize{[0.61, 2.62]}} & 5.79 {\scriptsize{[3.50, 11.00]}} & 15.29 {\scriptsize{[9.33, 25.41]}} \\ 
 BMDL$_q$ & 0.64 & 3.70 & 9.85\\ 
   \hline
\end{tabular}
\caption{For the three typical values of $q$, posterior mean and 95\% credible intervals of the BMD$_q$, along with the BMDL$_q$, characterizing the  5\% quantile of the posterior distribution of the BMD$_q$.}\label{tab:bmd}
\end{table}

We also performed inference varying $H$ to $H=5$ and $H=15$, but observed no relevant differences in inference on $R_{\mbox{\textsc{a}}}(x,37)$, BMD$_q$, and BMDL$_q$.

\section{Conclusion and Extensions}
\label{disc}

Motivated by quantitative risk assessment, we proposed a class of \comirep s balancing flexibility and parsimony in modeling conditional densities.  Although the focus is on situations in which the range of dose exposures is unbounded, \comire can be  applied also in studies when $x$ represents a concentration which is bounded above by $x_{\scriptsize{\mbox{max}}}$. In such settings, one appealing possibility is to focus on the rescaled predictor $\bar{x}=x/x_{\scriptsize{\mbox{max}}} \in [0,1]$ and define each  $\psi_j(\bar{x})$ in \eqref{eq:mixIsplines} via integrated Bernstein polynomials---i.e.  $\mbox{Beta}(j, 2^{\bar{J}} - j + 1)$---with $\bar{J}=\log_2\{J\}$. \comire can be also applied to binary health responses. In such settings it suffices to consider Bernoulli kernels $\mbox{K}(y; \btheta) =\pi^y(1-\pi)^{1-y}$, and let the extreme measures be $P_{\minx} =  \delta_{\pi_{\minpar}}$  and $P_{\maxx} =  \delta_{\pi_{\maxpar}}$, to obtain a fully flexible specification.

Consistent with our motivating application, we focus on the case in which lower outcome values $y$ are associated with adverse health. It is however straightforward to adapt the proposed model to scenarios in which higher values of $y$ are more adverse---e.g. blood cholesterol level. It is also possible to control for the effect of additional covariates $\bf{z}$, while maintaining the  adversity profile property in $x$. This can be done by allowing the location of each mixture component to change also as a function $g(\bf{z})$ of the additional covariates $\bf{z}$. Using the same $g(\bf{z})$ within each component, the adversity profile property in $x$ is maintained, and the additional risk  in $x$, for any $\bf{z}$, is $R_{\mbox{\textsc{a}}}(x,a;{\bf{z}})=\beta(x)R_{\mbox{\textsc{a}}}\{\infty,a ; g({\bf{z}})\}$, with  $R_{\mbox{\textsc{a}}}\{\infty, a ; g({\bf{z}})\}$ the additional risk at  $(\infty, \bf{z})$. Note also that, although we focus on observational data, \comire can be also applied in experimental studies where units are exposed to pre--specified doses. In such settings, \comire will efficiently interpolate between the fixed doses. In contrast, fully unstructured models, that attempt to infer the continuous dose--response relation from limited dose groups, will tend to have wide uncertainty across the gaps between the pre--specified doses.

It is also interesting to incorporate U-shape behaviors for $\mu(x)$ at low dose exposures, which may occur in toxicology studies \citep[e.g.][]{cala:2001}. One possibility to incorporate this property within \comire is to introduce an intermediate mixing distribution $P_{u}$ at the central value $x_u$ corresponding to the minimum of the U-shape. Under appropriate adversity profile restrictions for the atoms in $P_{\minx}$, $P_{u}$ and $P_{\maxx}$, the  U-shape can be incorporated, including uncertainty in $x_u$.

\appendix 
\section{Supplementary Material}
\subsection*{Formal Interpretation of the $\beta(x)$ Function}
According to Section 2.1  of the paper, \comire has several interpretations. Here, we formalize the metaphor interpreting the conditional density $f_x(y)$ as the result of a process which travels in distribution between the starting location $f_{\minx}(y)$ and the ending location $f_{\maxx}(y)$, as {\em time} $x$ increases from ${\minpar} $ to ${\maxpar}$. Theorem~\ref{teo1} clarifies the role of the $\beta(x)$ function in measuring the proportion of the path travelled at {\em time} $x$. 

\begin{theorem}
Let $F_{\minx}(y)$, $F_{\maxx}(y)$, and $F_x(y)$ characterize the probability distribution functions inducing the density functions $f_{\minx}(y)$, $f_{\maxx}(y)$ and $f_x(y)$, respectively, in equations (3)--(4), and let
\[
d_{\mbox{\textsc{tv}}}\{F_1(y),F_2(y)\}=\frac{1}{2}\int | f_1(y) - f_2(y) | d y\]
 be the total variation metric between the generic probability distribution functions $F_1(y)$ and $F_2(y)$ inducing the densities $f_1(y)$ and $f_2(y)$. Then 
 \[
 \beta(x) = d_{\mbox{\textsc{tv}}}\{F_x(y), F_{ \minx}(y)\}/d_{\mbox{\textsc{tv}}}\{F_{ \maxx}(y), F_{ \minx}(y)\}.
 \]
\label{teo1}
\end{theorem}

\begin{proof}
Theorem \ref{teo1}, can be easily derived after noticing that 
\begin{eqnarray*}
\int | f_x(y) - f_{\minx}(y) | d y 
& =& \int | \{1 - \beta(x)\} f_{\minx}(y) + \beta(x) f_{\maxx}(y) - f_{\minx}(y) | d y \\
&=&\beta(x) \int | f_{\maxx}(y) - f_{\minx}(y) | d y =  2 \beta(x) d_{\mbox{\textsc{tv}}}\{F_{\maxx}(y), F_{\minx}(y)\},
\end{eqnarray*}
and therefore   $\beta(x) = d_{\mbox{\textsc{tv}}}\{F_x(y), F_{ \minx}(y)\}/d_{\mbox{\textsc{tv}}}\{F_{ \maxx}(y), F_{ \minx}(y)\}$, thus proving Theorem~\ref{teo1}.
\end{proof}

\newpage

\subsection*{Additional Simulation Results}

 \begin{figure*}[t]
  \centering
     \subfigure[Scenario 1]{\includegraphics[width=.32\textwidth]{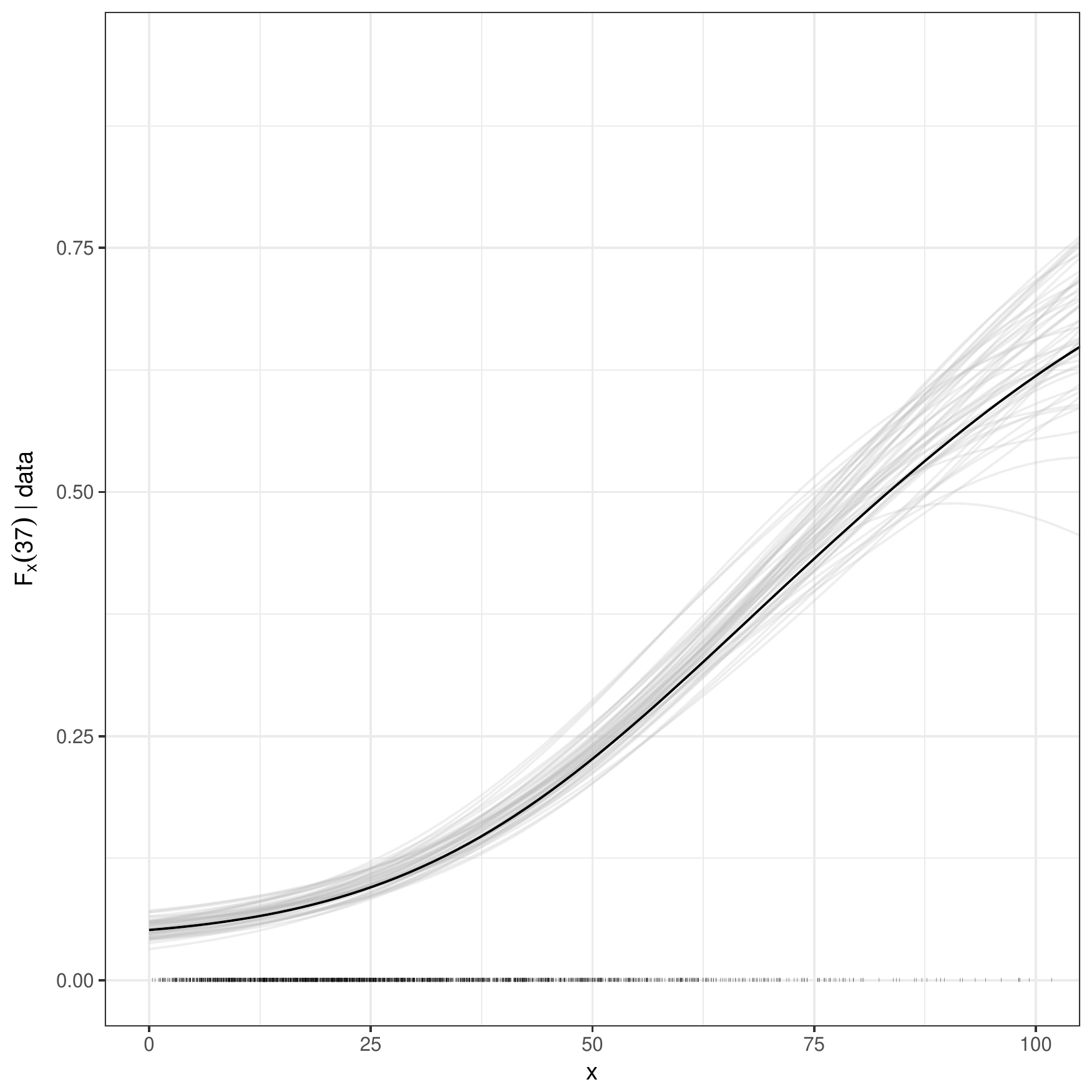}}
          \subfigure[Scenario 2]{\includegraphics[width=0.32\textwidth]{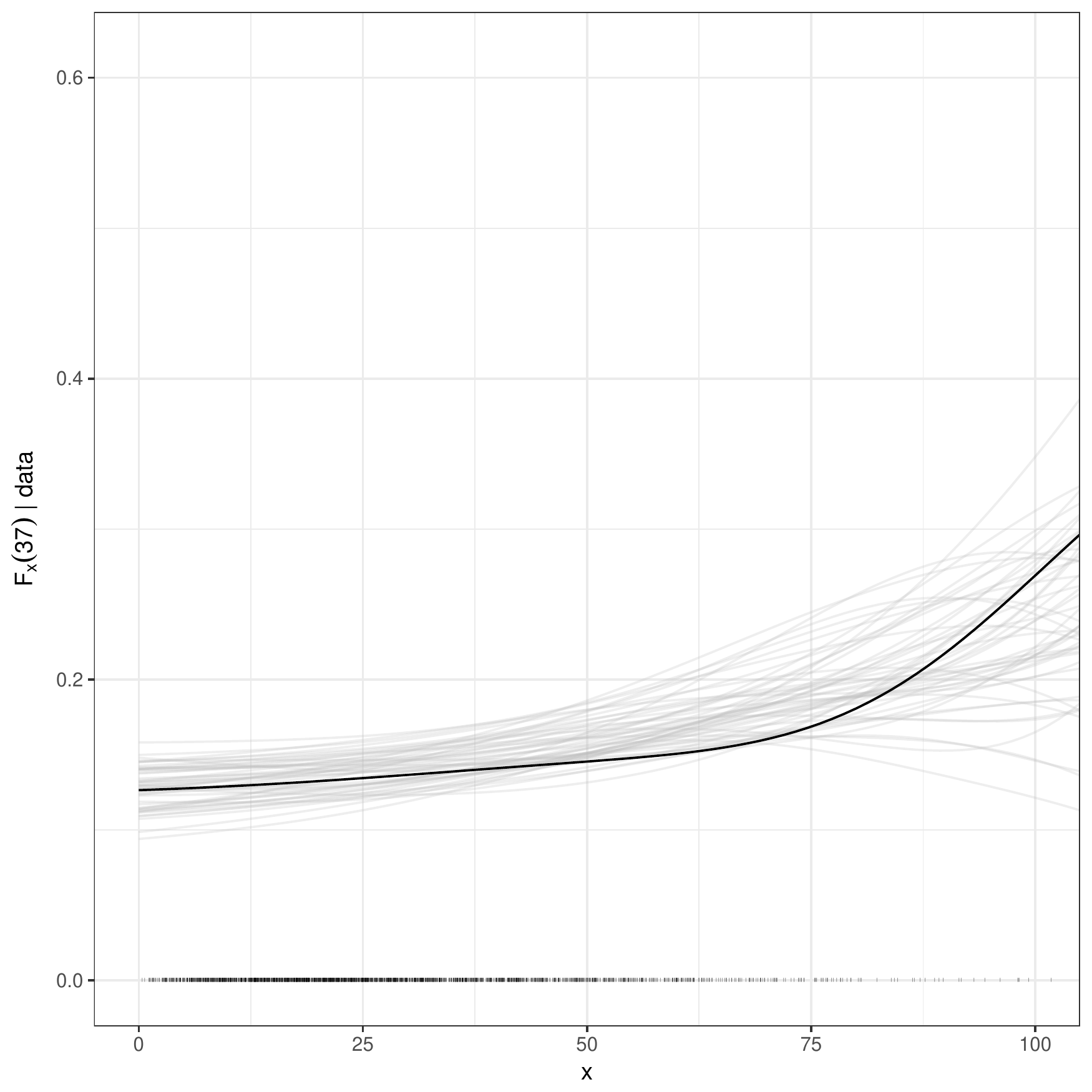}}
               \subfigure[Scenario 3]{\includegraphics[width=0.32\textwidth]{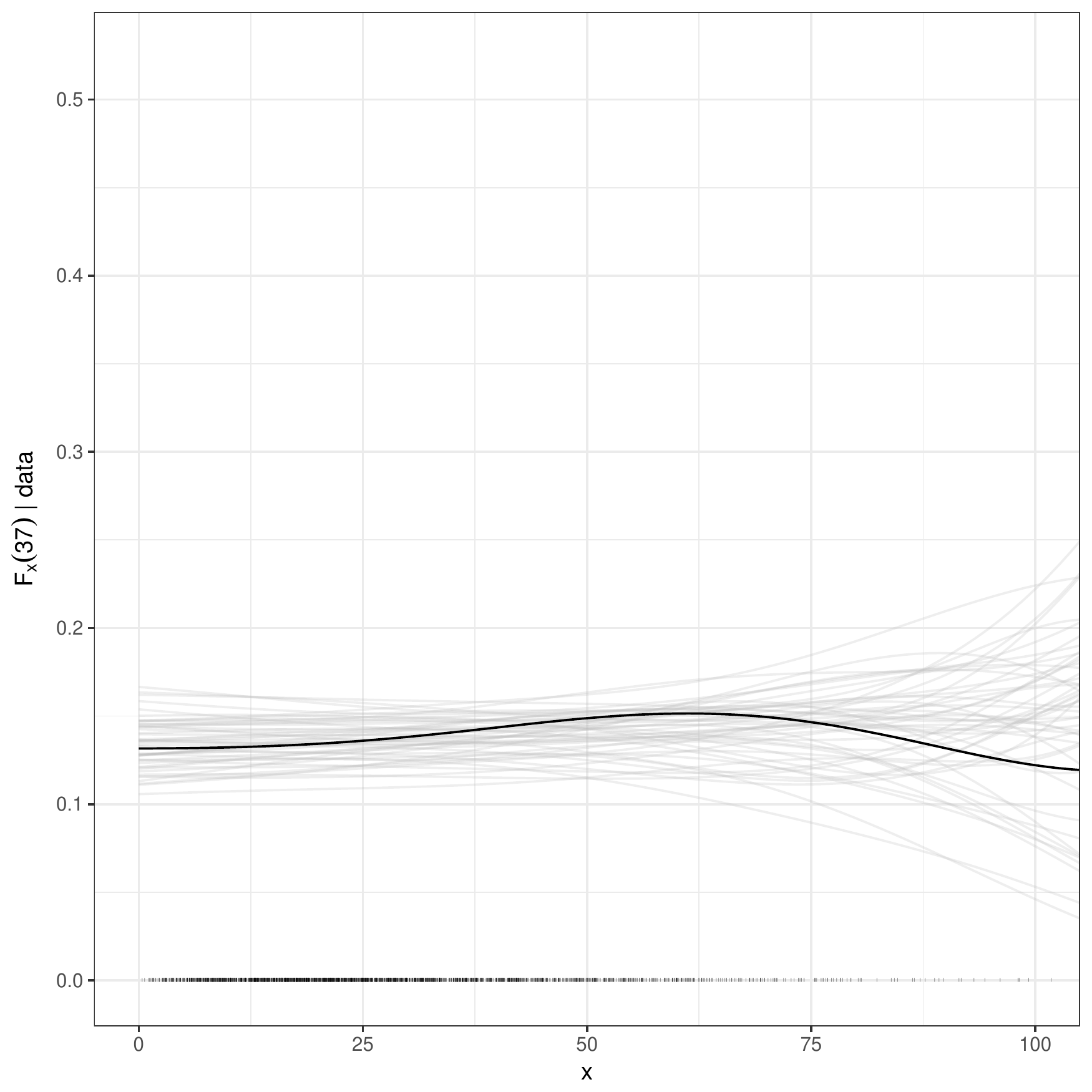}}
   \caption{Goodness-of-fit assessments for $n=2000$ in Scenario 1 (a), 2 (b), and 3  (c). Smoothed empirical estimate of $F_x(37)=\mbox{pr}(y \leq 37 \mid x)$ computed from the observed data (black line), and from $50$ data sets simulated from the posterior predictive distribution induced by  \comire (grey lines). In the $x$ axis we also report the simulated dose exposures.}
   \label{fig:postpredcheck_sm}
\end{figure*}

\begin{figure*}[t!]
  \centering
   \subfigure[Scenario 1]{\includegraphics[width=0.32\textwidth]{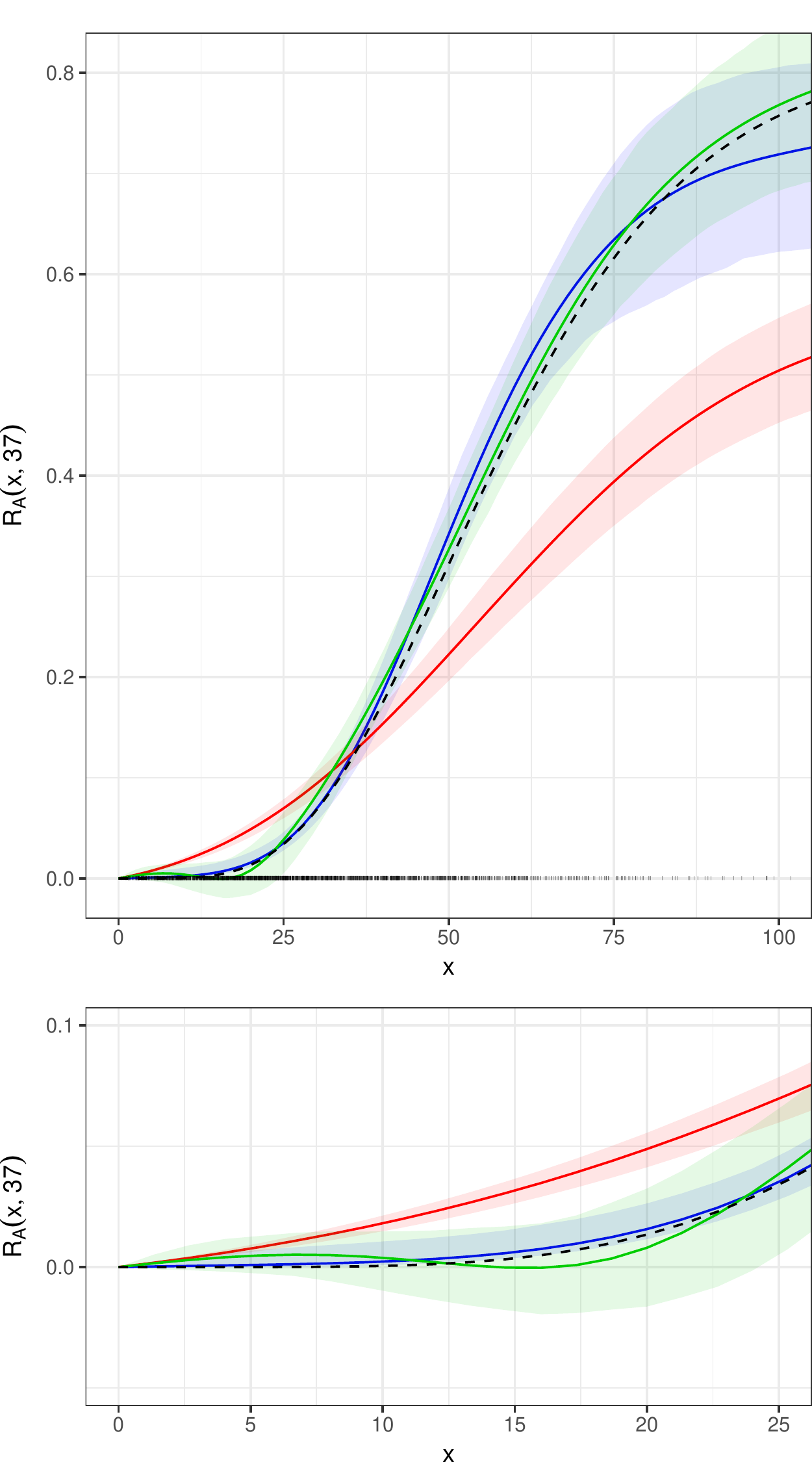}}
      \subfigure[Scenario 2]{\includegraphics[width=0.32\textwidth]{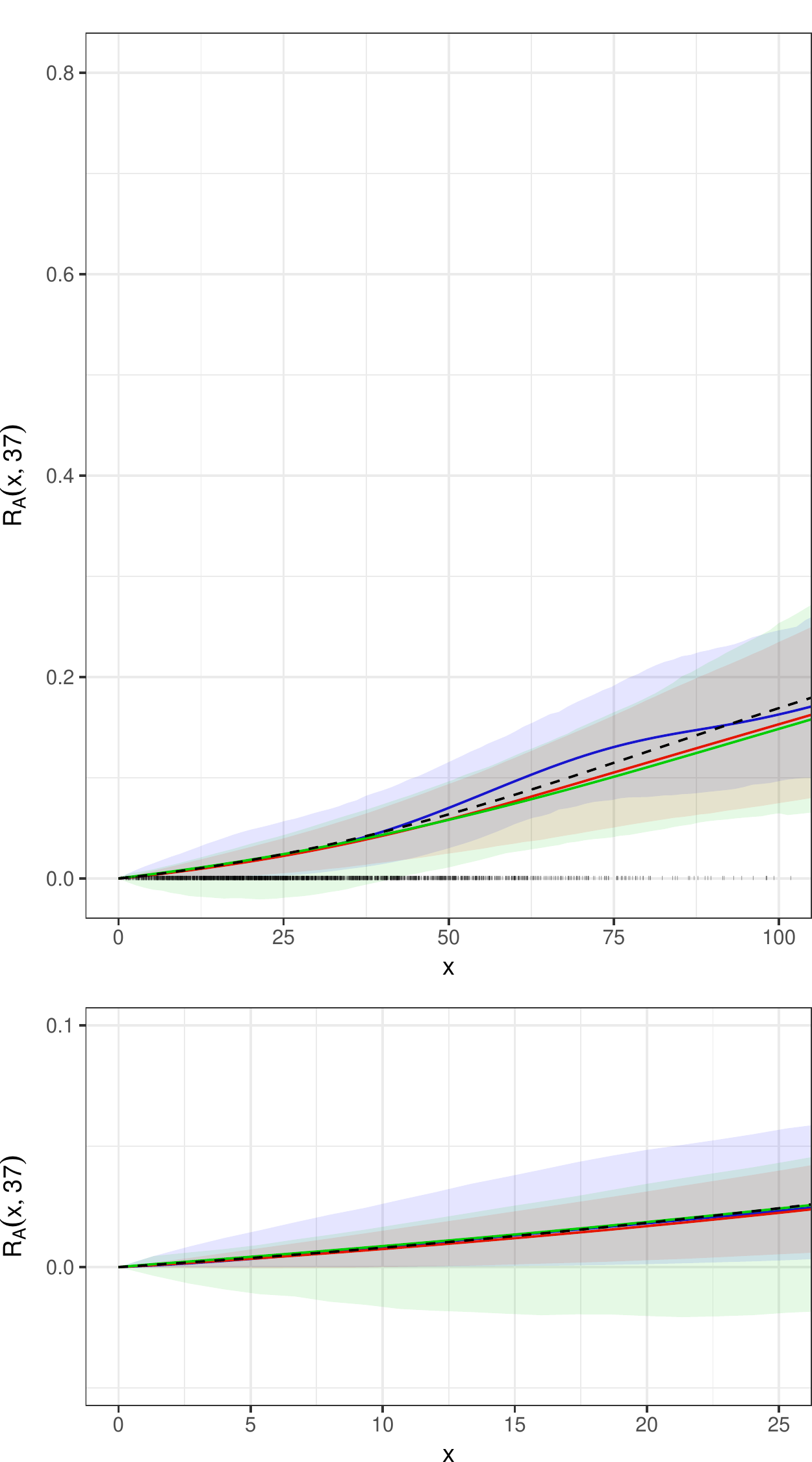}}
         \subfigure[Scenario 3]{\includegraphics[width=0.32\textwidth]{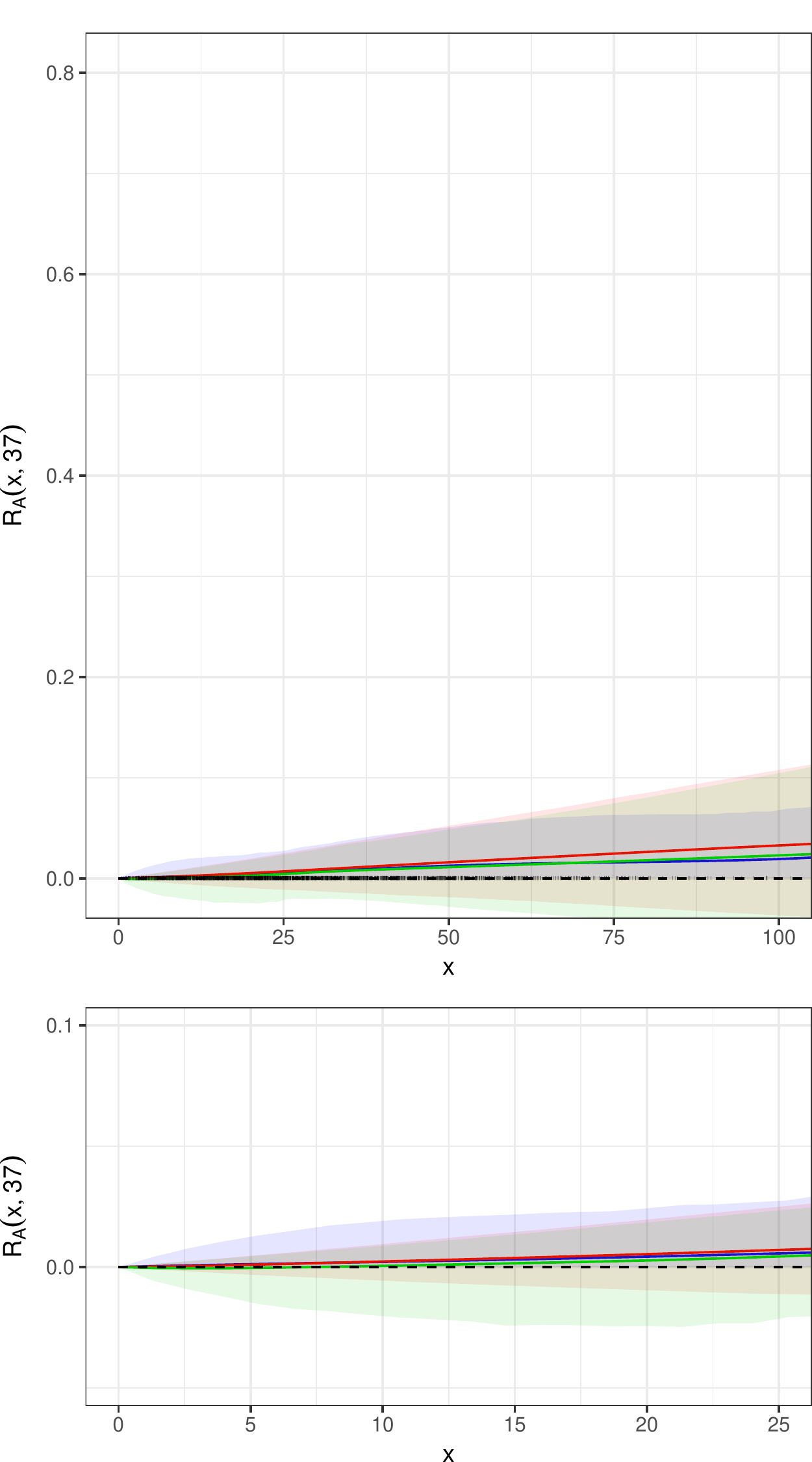}}
       \caption{Inference on $R_{\mbox{\textsc{a}}}(x,37)$ for $n=2000$ in the three scenarios.  The dashed lines represent the true additional risk function, whereas the red, green, and blue continuous lines denote the posterior mean of the additional risk under \textsc{anova--ddp}, \textsc{f--dmix}, and \comirep, respectively. The shaded areas represent the pointwise 95\% posterior credible bands. In the $x$ axis we report the simulated dose exposures. Lower panels provide a zoom on the range of the additional risk typically considered in benchmark dose analysis. }
   \label{fig:sim_inference1_sm}
\end{figure*}

Here, we reproduce the performance assessments of Section 4 in the paper under the same scenarios, simulation settings, prior specifications, and competing methods, but focusing on a larger sample  involving $n=2000$ units, instead of $n=500$. This situation mimics the sample size available in the CPP application. Refer to Section 4 for details on the three scenarios, prior specifications, and competing methods.

Consistent with the performance assessments in Section 4, we first check model adequacy in Figure \ref{fig:postpredcheck_sm} and then focus on evaluating inference for $R_{\mbox{\textsc{a}}}(x,a)$ in Figure~\ref{fig:sim_inference1_sm}. According to Figure \ref{fig:postpredcheck_sm}, also for $n=2000$, our approach  provides adequate fit to the smoothed empirical estimate of $F_x(37)=\mbox{pr}(y \leq 37 \mid x)$ based on the observed data in all the three scenarios. This motivates inference on the additional risk, whose quality is compared with the relevant competitors \textsc{anova--ddp} and \textsc{f--dmix} in Figure~\ref{fig:sim_inference1_sm}.

The results in Figure~\ref{fig:sim_inference1_sm} are in line with those discussed for  $n=500$ in Section 4 of the paper. Specifically,  \comire allows accurate learning of  $R_{\mbox{\textsc{a}}}(x,37)$ both in the correctly specified scenario, and in the cases of model misspecification.  The estimation of $R_{\mbox{\textsc{a}}}(x,37)$ is particularly precise at low--dose exposures, where quantitative risk assessments typically focus, as can be seen from the bottom panels of Figure~\ref{fig:sim_inference1_sm}. As expected, the increased simple size allows further improvements in posterior accuracy and precision for $R_{\mbox{\textsc{a}}}(x,37)$ compared to the situation in which $n=500$.

\subsection*{Code, Data and Tutorial Implementation}
Code, data and a tutorial implementation of \comire in the CPP application  are available at \url{github.com/tonycanale/CoMiRe}.


\section*{Acknowledgement}
The authors are grateful to the Editor, the Associate Editor and two referees for the helpful and constructive comments. The first author is supported by the University of Padova under the STARS Grant. This work is also partially supported by grant 1R01-ES028804 of the National Institute of Environmental Health Sciences of the US National Institutes of Health.  

\bibliographystyle{apalike}
{\fontsize{12}{14} \selectfont \bibliography{biblio}}

\label{lastpage}
\end{document}